\documentclass[noinfoline,stslayout]{imsart}

\usepackage{natbib}
\RequirePackage[colorlinks,citecolor=blue,urlcolor=blue]{hyperref}
\usepackage{url}
\usepackage{amssymb, mathtools}
\usepackage{comment}
\usepackage{xspace}
\usepackage{graphicx}
\usepackage{lmodern}
\usepackage{float}
\usepackage{xspace}
\usepackage{relsize}
\usepackage{amsmath}
\usepackage{amsthm}
\usepackage{mathrsfs}
\usepackage{multirow}
\usepackage{algorithm}
\usepackage{algorithmic}
\usepackage{booktabs} 
\usepackage{subfig}
\usepackage{ulem}

\newtheorem{prop}{Proposition}
\newtheorem{definition}{Definition}
\newtheorem{remark}{Remark}
\newtheorem{lemma}{Lemma}
\newtheorem{example}{Example}

\def\dvtx{\,:\,}

\newcommand{\R}{\mathbb{R}}
\newcommand{\Proba}{\mathbb{P}}

\def\d{\mathrm{d}}
\def\simind{\stackrel{\mathrm{ind}}{\sim}}
\def\simiid{\stackrel{\mathrm{i.i.d.}}{\sim}}

\newcommand{\Fcr}{\mathscr{F}}

\newcommand{\X}{\mathbb{X}}

\newcommand{\E}{\mathbb{E}}

\newcommand{\edr}{\mathrm{e}}
\newcommand{\ddr}{\mathrm{d}}
\newcommand{\Mc}{\mathscr{M}}
\newcommand{\tms}{\tilde{\mu}^{*}}

\newcommand{\bm}{\boldsymbol{m}}

\newcommand{\bJ}{\boldsymbol{J}}

\newcommand{\CRM}{CRM\xspace}

\newcommand{\cnkpar}{\binom{n}{k_1\cdots k_n}}
\newcommand{\cnkfact}{\frac{n!}{k_1!\ldots k_n!}}

\newcommand{\FK}{Ferguson \& Klass\xspace}
\newcommand{\FKa}{Ferguson \& Klass algorithm\xspace}

\newcommand{\BeP}{\ensuremath{\text{BeP}}\xspace}
\newcommand{\NGG}{normalized generalized gamma process\xspace}
\newcommand{\GG}{generalized gamma process\xspace}
\newcommand{\SBP}{stable-beta process\xspace}
\newcommand{\IBP}{Indian buffet process\xspace}
\newcommand{\IG}{inverse-Gaussian process\xspace}
\newcommand{\IGau}{inverse-Gaussian\xspace}
\newcommand{\NMC}{N_{\text{\tiny{FK}}}}
\newcommand{\EFK}{\mathbb{E}_{\text{\tiny{FK}}}}

\begin{document}

\begin{frontmatter}

\title{A moment-matching Ferguson \& Klass algorithm\protect\footnote{Julyan Arbel, Collegio Carlo Alberto, Via Real Collegio, 30, 10024  Moncalieri, Italy, \href{mailto:julyan.arbel@carloalberto.org}{julyan.arbel@carloalberto.org}\\ 
Igor Pr\"unster, Department of Decision Sciences, BIDSA and IGIER, Bocconi University, via Roentgen 1, 20136 Milan, Italy, \href{mailto:igor@unibocconi.it}{igor@unibocconi.it}. \\Research supported by the European Research Council (ERC) through StG ``N-BNP'' 306406. }
}

\author{Julyan Arbel$^{1,2}$ \and Igor Pr\"unster$^{1}$}
%
\affiliation{$^{1}$Bocconi University, Milan, Italy}
\affiliation{$^{2}$Collegio Carlo Alberto, Moncalieri, Italy}

\maketitle

\begin{abstract}
Completely random measures (CRM) represent the key building block of a wide variety of popular stochastic models and play a pivotal role in modern Bayesian Nonparametrics. A popular representation of CRMs as a random series with decreasing jumps is due to \cite{ferguson1972representation}. This can immediately be turned into an algorithm for sampling realizations of CRMs or more elaborate models involving transformed CRMs. However, concrete implementation requires to truncate the random series at some threshold resulting in an approximation error. The goal of this paper is to quantify the quality of the approximation by a moment-matching criterion, which consists in evaluating a measure of discrepancy between actual moments and moments based on the simulation output. Seen as a function of the truncation level, the methodology can be used to determine the truncation level needed to reach a certain level of precision. The resulting moment-matching \FK algorithm is then implemented and illustrated on several popular Bayesian nonparametric models.
\end{abstract}


\end{frontmatter}

\section{Introduction}\label{sec:introduction}
Independent increment processes or, more generally, completely random measures (CRMs) are ubiquitous in modern stochastic modeling and inference. They form the basic building block of countless popular models in, e.g., Finance, Biology, Reliability, Survival Analysis. Within Bayesian nonparametric statistics they play a pivotal role. The Dirichlet process, the cornerstone of the discipline introduced in \cite{ferguson1973bayesian}, can be obtained as normalization or exponentiation of suitable CRMs \cite[see][]{Ferguson1974}.
 Moreover, as shown in \cite{lijoi2010models}, CRMs can be seen as the unifying concept of a wide variety of Bayesian nonparametric models. See also  \cite{Jordan2010hierarchical}. The concrete implementation of models based on CRMs often requires to simulate their realizations. Given they are discrete infinite objects, $\sum_{i \geq 1} J_i \delta_{Z_i}$, some kind of truncation is required, producing an approximation error $\sum_{i \geq M+1} J_i \delta_{Z_i}$. Among the various representations useful for simulating realizations of CRMs the method  due to \cite{ferguson1972representation} and popularized by \cite{Walker2000representations} stands out in that, for each realization, the weights $J_i$'s are sampled in decreasing order. This clearly implies that for a given truncation level $M$ the approximation error over the whole sample space is minimized. The appealing feature of decreasing jumps has lead to a huge literature exploiting the Ferguson \& Klass algorithm. Limiting ourselves to recall contributions within Bayesian Nonparametrics we mention, among others,
\citet{argiento2015blocked,argiento2015priori,barrios2013modeling,DeBlasi2010class,epifani2003exponential,Griffin2011posterior,griffin2016adaptive,Nieto-Barajas2002markov,nieto2004normalized,Nieto-Barajas2004bayesian,nieto2009sensitivity,Nieto-Barajas2014bayesian}.
General references dealing with the simulation of L\'evy processes include \citet{rosinski2001series} and \citet{tankov2008financial}, who review the \FKa and the compound Poisson process approximation to a L\'evy process.

However, the assessment of the quality of the approximation due to the truncation for general CRMs is limited to some heuristic criteria. For instance, \citet{barrios2013modeling} implement the \FKa for mixture models by using the so called \textit{relative error} index. The corresponding stopping rule prescribes to truncate when the relative size of an additional jump is below a pre-specified fraction of the sum of sampled jumps. The inherent drawbacks of such a procedure and related heuristic threshold-type procedures employed in the several of the above references is two-fold. On the one hand the threshold is clearly arbitrary without quantifying the total mass of the ignored jumps. On the other hand the total mass of the jumps beyond the threshold, i.e. the approximation error, can be very different for different CRMs or, even, for the same CRM with different parameter values; this implies that the same threshold can produce very different approximation errors in different situations. Starting from similar concerns about the quality of the approximation, the recent paper by \citet{griffin2016adaptive} adopts an algorithmic approach and proposes an adaptive truncation sampler based on sequential Monte Carlo for infinite mixture models based on normalized random measures and on stick-breaking priors. The measure of discrepancy that is used in order to assess the convergence of the sampler is based on the effective sample size (ESS) calculated over the set of particles: the algorithm is run until the absolute value of the difference between two consecutive ESS gets under a pre-specified threshold. Also motivated by the same concerns, \citet{argiento2015blocked,argiento2015priori} adopt an interesting approach to circumvent the problem of truncation by changing the model in the sense of replacing the CRM part of their model with a Poisson process approximation, which having an (almost surely) finite number of jumps can be sampled exactly. However, this leaves the question of the determination of the quality of approximation for truncated CRMs open.
Another line of research, originated by \citet{ishwaran2001gibbs}, is dedicated to validating the trajectories from the point of view of the marginal density of the observations in mixture models. In this context, the quality of the approximation is measured by the $L_1$ distance between the marginal densities under truncated and non-truncated priors. Recent interesting contributions in this direction include bounds for a Ferguson \& Klass representation of the beta process \citep{doshi2009variational} and  bounds for the beta process, the Dirichlet process as well as for arbitrary CRMs in a \textit{size biased representation} \citep{paisley2012stick,brod}.

This paper faces the problem by a simple yet effective idea. In contrast to the above strategies, our approach takes all jumps of the CRMs into account and hence leads to select truncation levels in a principled way, which vary according to the type of CRM and its parameters. The idea is as follows: given moments of CRMs are simple to compute, one can quantify the quality of the approximation by evaluating some measure of discrepancy between the actual moments of the CRM at issue (which involve all its jumps) and the ``empirical'' moments, i.e. the moments computed based on the truncated sampled realizations of the CRM. By imposing such a measure of discrepancy not to exceed a given threshold and selecting the truncation level $M$ large enough to achieve the desired bound, one then obtains a set of ``validated'' realizations of the CRM, or, in other terms, satisfying a moment-matching criterion. An important point to stress is that our validation criterion is all-purpose in spirit since it aims at validating the CRM samples themselves rather than samples of a transformation of the CRM. Clearly the latter type of validation would be ad hoc, since it would depend on the specific model. For instance, with the very same set of moment-matching realizations of a gamma process, one could obtain a set of realizations of the Dirichlet process via normalization and a set gamma mixture hazards by combination with a suitable kernel. Moreover, given moments of transformed CRMs are typically challenging to derive, a moment-matching strategy would not be possible in most cases. Hence, while the quantification of the approximation error does not automatically translate to transformed CRMs, one can still be confident that the moment-matching output at the CRM level produces good approximations. That this is indeed the case is explicitly shown in some practical examples both for prior and posterior quantities in Section \ref{sec:BNP}.

The outline of the paper is as follows. In Sections \ref{sec:CRM_properties}-\ref{sec:moments} we recall the main properties of CRMs and provide expressions for their moments. In Sections \ref{sec:FK}-\ref{sec:moment_match} we describe the Ferguson \& Klass algorithm and introduce the measure of discrepancy between moments used to quantify the approximation error due to truncation. Section \ref{sec:BNP} illustrates the moment-matching Ferguson \& Klass algorithm for some popular CRMs and CRM-based Bayesian nonparametric models, namely normalized CRMs and the beta-stable Indian buffet process. Some probabilistic results, discussed in Section \ref{sec:FK}, are given in the Appendix.

\section{Completely random measures}\label{sec:CRM}

\subsection{Definition and main properties}\label{sec:CRM_properties}

Let $\Mc_\X$ be the set of boundedly finite measures on $\X$, which means that if $\mu \in \Mc_\X$ then $\mu(A)<\infty$ for any bounded set $A$. $\X$ is assumed to be a complete and separable metric space and both $\X$ and $\Mc_\X$ are equipped with the corresponding Borel $\sigma$-algebras. See \citet{daley2008introduction} for details.

\begin{definition}
A random element $\tilde\mu$, defined on $(\Omega,\Fcr,\mathbb{P})$ and taking values in $\Mc_\X$, is called a \emph{completely random measure} (CRM) if, for any collection of pairwise disjoint sets $A_1,\ldots,A_n$ in $\X$, the random variables $\tilde\mu(A_1),\ldots,\tilde\mu(A_n)$ are mutually independent.
\end{definition}

An important feature is that a CRM $\tilde \mu$ selects (almost surely) discrete measures and hence can be represented as
    \begin{equation} \label{eq:CRM_discrete}
    \tilde\mu=\sum_{i\ge1}J_i \delta_{Z_i}
    \end{equation}
where the jumps $J_i$'s and locations $Z_i$'s are random and independent. In \eqref{eq:CRM_discrete} and throughout we assume there are no fixed points of discontinuity a priori. The main technical tool for dealing with CRMs is given by their Laplace transform, which admits a simple structural form known as L\'evy--Khintchine representation. In fact, the Laplace transform of $\tilde\mu(A)$, for any $A$ in $\X$, is given by
    \begin{equation}\label{eq:lapl_transform}
    L_A(u) = \E \bigl[\edr^{-\lambda\tilde\mu(A)} \bigr] =\exp \biggl\{- \int_{\R^+\times A} \bigl[1-\edr^{- \lambda v } \bigr] \nu(
    \ddr v,\ddr x) \biggr\}
    \end{equation}
for any $\lambda>0$. The measure $\nu$ is known as \emph{L\'evy intensity} and uniquely characterizes $\tilde \mu$. In particular, there corresponds a unique CRM $\tilde \mu$ to any measure $\nu$ on $\R^+\times\X$ satisfying the integrability condition
    \begin{equation}
    \label{eq:levy_intensity}
    \int_B\int_{\R^+}\min\{v,1\} \nu(\ddr v,\ddr x)<\infty
    \end{equation}
for any bounded $B$ in $\X$. From an operational point of view this is extremely useful, since a single measure $\nu$ encodes all the information about the jumps $J_i$'s and the locations $Z_i$'s. The measure $\nu$ will be conveniently rewritten as
\begin{equation}
\label{pim} \nu(\ddr v,\ddr x)=\rho(\ddr v| x) \alpha(\ddr x),
\end{equation}
where $\rho$ is a transition kernel on $\R^+ \times\X$ controlling the jump intensity and $\alpha$ is a measure on $\X$ determining the
locations of the jumps. If $\rho$ does not depend on $x$, the CRM is said homogeneous, otherwise it is non-homogeneous.

We now introduce two popular examples of CRMs that we will serve as illustrations throughout the paper.

\begin{example}\label{ex:GG}
The \textit{\GG} introduced by \citet{brix1999generalized} is characterized by a L\'evy intensity of the form
\begin{equation}
\label{eq:ngg} \nu(\d v, \d x) =\frac{\edr^{-\theta v}}{\Gamma(1-\gamma)  v^{1+\gamma}}\,\d v\, \alpha(\d x),
\end{equation}
whose parameters $\theta\geq0$ and $\gamma\in[0,1)$ are such that at least one of them is strictly positive.
Notable special cases are: (i) the \emph{gamma} CRM which is obtained by setting $\gamma=0$; (ii)~the \IGau CRM, which arises by fixing $\gamma=0.5$; (iii) the \emph{stable} CRM which corresponds to $\theta=0$. Moreover, such a CRM stands out for its analytical tractability. In the following we work with $\theta=1$, a choice which excludes the stable CRM. This is justified in our setting because the moments of the stable process do not exist. See Remark \ref{rem:stable}.
\end{example}

\begin{example}\label{ex:SBP}
The \textit{stable-beta process}, or \textit{three-parameter beta process}, was defined by \citet{teh2009indian} as an extension of the beta process
\citep{hjort1990nonparametric}. Its jump sizes are upper-bounded by $1$ and its L\'evy intensity on $[0,1]\times\X$ is given by
    \begin{equation}
    \label{eq:stable_beta_intensity}
    \nu(\d v, \d x) =\frac{\Gamma(c+1)}{\Gamma(1-\sigma)\Gamma(c+\sigma)} v^{-\sigma-1}(1-v)^{c+\sigma-1}\d v\, \alpha(\d x),
    \end{equation}
where $\sigma\in[0,1)$ is termed \textit{discount parameter} and $c>-\sigma$ \textit{concentration parameter}. When $\sigma=0$, the \SBP reduces to the beta CRM of \cite{hjort1990nonparametric}. Moreover, if $c=1-\sigma$, it boils down to a stable CRM where the jumps larger than 1 are discarded.
\end{example}

\subsection{Moments of a CRM\label{sec:moments}}

For any measurable set $A$ of $\X$, the $n$-th (raw) moment of $\tilde \mu(A)$ is defined by
$$m_n(A) = \E\big[\tilde \mu^n (A)\big].$$ In the sequel the multinomial coefficient is denoted by $\cnkpar = \cnkfact$.
In the next proposition we collect known results about moments of CRMs which are crucial for our methodology.
\begin{prop}\label{prop:crm_moments}
Let $\tilde \mu$ be a CRM with L\'evy intensity $\nu(\d v, \d x)$. Then the $i$-th cumulant of $\tilde \mu(A)$, denoted by $\kappa_i(A)$, is given by
$$\kappa_i(A)=\int_{\R^+\times A}v^i\nu(\d v, \d x),$$
which, in the homogeneous case $\nu(\d v, \d x) = \rho(\d v)\alpha(\d x)$, simplifies to
$$\kappa_i(A) = \alpha(A)\int_0^\infty v^i\rho(\d v).$$
The $n$-th moment of $\tilde\mu(A)$ is given by
\begin{equation*}
m_n(A) = \sum_{(*)} {\scriptstyle\cnkpar} \prod_{i=1}^n \big(\kappa_i(A)/i!\big)^{k_i},
\end{equation*}
where the sum $(*)$ is over all $n$-tuples of nonnegative integers $(k_1,\ldots,k_n)$ satisfying the constraint $k_1+2k_2+\cdots+nk_n=n$.
\end{prop}
A proof is given in the Appendix.

\smallskip

In the following we focus on (almost surely) finite CRMs i.e. $\tilde \mu(\X)<\infty$. This is motivated by the fact that most Bayesian nonparametric models, but also models in other application areas, involve finite CRMs. Hence, we assume that the measure $\alpha$ in \eqref{eq:levy_intensity} is finite i.e. $\alpha(\X)\coloneqq a \in (0,\infty)$. This is a sufficient condition for $\tilde \mu(\X)<\infty$ in the non-homogeneous case and also necessary in the homogeneous case \citep[see e.g.][]{rlp2003}. A common useful parametrization of $\alpha$ is then given as $a P^*$ with $P^*$ a probability measure and $a$ a finite constant. Note that, if $\tilde \mu(\X)=\infty$, one could still identify a bounded set of interest $A$ and the whole following analysis carries over by replacing $\tilde \mu(\X)$ with $\tilde \mu(A)$.

\smallskip

As we shall see in Section~\ref{sec:FK}, the key quantity for evaluating the truncation error is given by the random total mass of the CRM, $\tilde \mu(\X)$. Proposition~\ref{prop:crm_moments} shows how the moments $m_n=m_n(\X)$ can be obtained from the cumulants $\kappa_i=\kappa_i(\X)$ and, in particular, the relations between the first four moments and the cumulants are
    \begin{align*}
    m_1 &= \kappa_1, \,
    m_2= \kappa_1^2 + \kappa_2, \,
    m_3=  \kappa_1^3 + 3\kappa_1\kappa_2 + \kappa_3, \,
    m_4  = \kappa_1^4 + 6 \kappa_1^2\kappa_2 + 4\kappa_1\kappa_3 + 3\kappa_2^2 + \kappa_4.
    \end{align*}
With reference to the two examples considered in Section \ref{sec:CRM_properties}, in both cases the expected value of $\tilde \mu(\X)$ is $a$, which explains the typical terminology \emph{total mass parameter} attributed to $a$. For the generalized gamma CRM the variance is given by $\text{Var}(\tilde\mu(\X)) = a(1-\gamma)$, which shows how the parameter $\gamma$ affects the variability. Moreover, $\kappa_i = a(1-\gamma)_{(i-1)}$ with $x_{(k)}=x(x+1)\ldots(x+k-1)$ denoting the ascending factorial. As for the stable-beta CRM, we have $\text{Var}(\tilde\mu(\X)) = a\frac{1-\sigma}{c+1}$ with both discount and concentration parameter affecting the variability, and also $\kappa_i = a\frac{(1-\sigma)_{(i-1)}}{(1+c)_{(i-1)}}$. Table~\ref{tab:moments} summarizes the cumulants $\kappa_i$ and moments $m_n$ for the random total mass $\tilde \mu(\X)$  for the generalized gamma (assuming as in Example \ref{ex:GG} $\theta=1$), stable-beta CRMs and some of their special cases.

\renewcommand{\arraystretch}{1.8}

\begin{table}
\begin{tabular}{|c|c|ccccccc|}
\hline
\CRM & Cumulants & \multicolumn{7}{ c |}{Moments} \\ \hline
& $\kappa_i$ & $m_1$& & $m_2$ && $m_3$ && $m_4$ \\
\hline
G & $a(i-1)!$ & $a$ && $a_{(2)}$ && $a_{(3)}$ && $a_{(4)}$ \\
\hline
IG & $a{(1/2)_{(i-1)}}$ & ${a}$ && ${a^2}+\frac{1}{2}a$ && ${a^3}+\frac{3}{2}a^2 $ && ${a^4}+{3a^3}$ \\
&&&&&&$+\frac{3}{4}a$&& $+\frac{15}{4}a^2+\frac{15}{8}a$\\
\hline
GG & $a{(1-\gamma)_{(i-1)}}$ & ${a}$ && ${a^2}+{a{\scriptstyle(1-\gamma)}}$ && ${a^3}+3a^2{\scriptstyle(1-\gamma)}$ && ${a^4}+6a^3{\scriptstyle(1-\gamma)}$ \\
&&&&&&$+{a{\scriptstyle(1-\gamma)_{(2)}}}$ && $+a^2{\scriptstyle{(1-\gamma)(11-7\gamma)}}+{a{\scriptstyle(1-\gamma)_{(3)}}}$\\
\hline
B & $a\frac{(i-1)!}{(c+1)_{(i-1)}}$ & $a$ && $a^2+\frac{a}{c+1}$ && $a^3+\frac{3a^2}{c+1} $ & &$a^4+\frac{6a^3}{c+1} + \frac{8a^2}{(c+1)_{(2)}} $\\
&&&&&&$+ \frac{2a}{(c+1)_{(2)}}$ && $+\frac{3a^2}{(c+1)^2} + \frac{6a}{(c+1)_{(3)}}$\\
\hline
SB & $a\frac{(1-\sigma)_{(i-1)}}{(c+1)_{(i-1)}}$ & $a$ & &$a^2+a\frac{1-\sigma}{c+1}$ && $a^3+3a^2\frac{1-\sigma}{c+1} $ && $a^4+6a^3\frac{1-\sigma}{c+1} + 4a^2\frac{(1-\sigma)_{(2)}}{(c+1)_{(2)}}$ \\
&&&&&&$+ a\frac{(1-\sigma)_{(2)}}{(c+1)_{(2)}}$ && $+ 3a^2\frac{(1-\sigma)^2}{(c+1)^2}+a\frac{(1-\sigma)_{(3)}}{(c+1)_{(3)}}$\\
\hline
\end{tabular}
\caption{\label{tab:moments}
Cumulants and first four moments  of the random total mass $\tilde \mu(\X)$ for the gamma (G), \IGau (IG), generalized gamma (GG), beta (B) and stable-beta (SB) CRMs.}
\end{table}

\begin{remark} \label{rem:stable}
{\rm The \textit{stable} CRM, which can be derived from the generalized gamma CRM by setting $\theta=0$, does not admit moments. Hence, it cannot be included in our moment-matching methodology. However, the \textit{stable} CRM with jumps larger than $1$ discarded, derived from the \SBP by setting $c=1-\sigma$, has all moments. Moreover, even when working with the standard stable CRM, posterior quantities typically involve an exponential updating of the L\'evy intensity \citep[see][]{lijoi2010models}, which makes the corresponding moments finite. This then allows to apply the moment matching methodology to the posterior.}
\end{remark}

\subsection{\FK algorithm\label{sec:FK}}

For notational simplicity we present the \FK algorithm for the case $\X = \R$. However, note that it can be readily extended to more general Euclidean spaces \citep[see e.g.][]{orbanz2011unit}.
Given a CRM
\begin{equation}
\tilde \mu = \sum_{i=1}^\infty J_i \delta_{Z_i},
\label{eq:mu_FK}
\end{equation}
the \FK representation consists in expressing random jumps $J_i$ occurring at random locations $Z_i$ in terms of the underlying L\'evy intensity.

\noindent In particular, the random locations $Z_i$, conditional on the jump sizes $J_i$, are obtained from the distribution function $F_{Z_i|J_i}$ given by
\[
F_{Z_i|J_i}(s)=\frac{\nu(\ddr J_i , (-\infty,s])}{\nu(
\ddr J_i,\R)}.
\]
In the case of a homogeneous CRM with L\'evy intensity $\nu(\d v, \d x) =\rho(\d v)\, a P^*(\d x)$, the jumps are independent of the locations and, therefore $F_{Z_i|J_i}=F_{Z_i}$ implying that the locations are i.i.d. samples from $P^*$.

\noindent As far as the random jumps are concerned, the representation produces them in decreasing order, that is, $J_1\geq J_2\geq \cdots $. Indeed, they are obtained as $\xi_i=N(J_i)$, where $N(v)=\nu([v,\infty),\R)$ is a decreasing function, and $\xi_1,\xi_2,\ldots$ are jump times of a
standard Poisson process (PP) of unit rate i.e. $\xi_1,\xi_2-\xi_1,\ldots\simiid\text{Exp}(1)$.
Therefore, the $J_i$'s are obtained by solving the equations $\xi_i=N(J_i)$. In general, this is achieved by numerical integration, e.g., relying on quadrature methods \citep[see, e.g.][]{burden1993numerical}. For specific choices of the CRM, it is possible to make the equations explicit or at least straightforward to evaluate. For instance, if $\tilde \mu$ is a \GG (see Example \ref{ex:GG}), the function $N$ takes the form
\begin{equation}
\label{funcM}  N(v) = \frac{a}{\Gamma(1-\gamma)}\int_v^\infty
\edr^{- u}u^{-(1+\gamma)}\,\d u =  \frac{a}{\Gamma(1-\gamma)}\Gamma(v; -\gamma),
\end{equation}
with $\Gamma(\,\cdot\,;\,\cdot\,)$ indicating an incomplete gamma function. If $\tilde \mu$ is the stable-beta process, one has
\begin{equation}
\label{eq:funcM_BP}  N(v) =a\frac{\Gamma(c+1)}{\Gamma(1-\sigma)\Gamma(c+\sigma)}\int_v^1u^{-\sigma-1}(1-u)^{c+\sigma-1}\,\d u =
a\frac{\Gamma(c+1)}{\Gamma(1-\sigma)\Gamma(c+\sigma)}B(1-v;c+\sigma,-\sigma),
\end{equation}
where $B(\,\cdot\,;\,\cdot\,,\,\cdot\,)$ denotes the incomplete beta function.

\medskip

Hence, the \FK algorithm can be summarized as follows.\\[-1cm]
\begin{center}
\begin{minipage}[c]{6.3cm}
\begin{algorithm}[H]
\caption{\FKa\label{algo:FKa}}
  \begin{algorithmic}[1]
  \STATE Sample $\xi_i\sim\text{PP}$ for $i=1,\ldots,M$
  \STATE Define $J_i=N^{-1}(\xi_i)$ for $i=1,\ldots,M$
  \STATE Sample $Z_i\sim P^*$ for $i=1,\ldots,M$
  \STATE Approximate $\tilde\mu$ by $\sum_{i=1}^M J_i \delta_{Z_i} $
   \end{algorithmic}
\end{algorithm}
\end{minipage}
\end{center}

\medskip

Since it is impossible to sample an infinite number of jumps, approximate simulation of $\tilde \mu$ is in order. This becomes a question of determining the number $M$ of jumps to sample in~\eqref{eq:mu_FK} leading to the truncation
\begin{equation}\label{eq:mu_approx}
\tilde \mu \approx \tilde \mu_M = \sum_{i=1}^M J_i \delta_{Z_i},
\end{equation}
with approximation error in terms of the un-sampled jumps equal to $\sum_{i=M+1}^\infty J_i$. The \FK representation has the key advantage of generating the jumps in decreasing order implicitly minimizing such an approximation error. Then, the natural path to determining the truncation level $M$ would be the evaluation of the \FK tail sum
\begin{equation}\label{eq:tail_sum}
\sum_{i=M+1}^\infty N^{-1}(\xi_i).
\end{equation}

\citet[][Theorem A.1]{brix1999generalized} provided an upper bound for \eqref{eq:tail_sum} in the generalized gamma case. In Proposition \ref{prop:proba} of Appendix~\ref{sec:app_tail} we derive also an upper bound for the tail sum of the stable-beta process. However, both bounds are far from sharp and therefore of little practical use as highlighted in Appendix \ref{sec:app_tail}. This motivates the idea of looking for a different route and our proposal consists in the  moment-matching technique detailed in the next section.

\subsection{Moment-matching criterion}\label{sec:moment_match}

Our methodology for assessing the quality of approximation of the \FKa consists in comparing the actual distribution of the random total mass $\tilde \mu(\X)$ with its empirical counterpart, where by empirical distribution we mean the distribution obtained by the sampled trajectories, i.e. by replacing random quantities by Monte Carlo averages of their sampled trajectories. In particular, based on the fact that the first $K$ moments carry much information about a distribution, theoretical and empirical moments of $\tilde \mu(\X)$ are compared.

The infinite vector of jumps is denoted by $\bJ=(J_i)_{i=1}^\infty$ and a vector of jumps sampled by the \FKa by $\bJ^{(l)}=(J_1^{(l)},\ldots,J_M^{(l)})$. Here, $l=1,\ldots,\NMC$ stands for the $l$-th iteration of the algorithm, i.e. for the $l$-th sampled realization. We then approximate the expectation $\E$ of a statistic of the jumps, say $S(\bJ)$, by the following empirical counterpart, denoted by $\EFK$,
\begin{equation}\label{eq:approx_gen}
\E\big[S(\bJ)]\approx \EFK\big[S(\bJ)] \coloneqq \frac{1}{\NMC}\sum_{l=1}^{\NMC} S\big(\bJ^{(l)}\big).
\end{equation}
Note that there are two layers of approximation involved in \eqref{eq:approx_gen}: first, only a finite number of jumps $M$ is used; second, the actual expected value is estimated through an empirical average which typically conveys on Monte Carlo error. The latter is not the focus of the paper, so we take a large enough number of trajectories, $\NMC = 10^4$, in order to insure a limited Monte Carlo error of the order of $0.01$. We focus on the first approximation inherent to the \FKa .

More specifically, as far as moments are concerned, $\bm_K = (m_1,\ldots,m_K)$ denotes the first $K$ moments of the random total mass $\tilde \mu(\X)=\sum_{i=1}^\infty J_i$ provided in Section \ref{sec:moments} and $\hat \bm_{K} = (\hat m_1,\ldots,\hat m_K)$ indicates the first $K$ empirical moments given by
\begin{equation}\label{eq:approx_moments}
\hat m_n = \EFK\left[\left(\sum_{i=1}^M J_i\right)^n\right].
\end{equation}

As measure of discrepancy between theoretical and empirical moments, a natural choice is given by the mean squared error between the vectors of moments or, more precisely, between the $n$-th roots of theoretical and empirical moments
\begin{equation}\label{eq:MSE}
\ell = \ell(\bm_K,\hat \bm_K) = \bigg(\frac{1}{K}\sum_{n=1}^K\big(m_n^{1/n}-\hat m_n^{1/n}\big)^2\bigg)^{1/2}.
\end{equation}

When using the \FK representation for computing the empirical moments the index $\ell$ depends on the truncation level $M$ and we highlight such a dependence by using the notation $\ell_M$. Of great importance is also a related quantity, namely the number of jumps necessary for achieving a given level of precision, which essentially consists in inverting $\ell_M$ and is consequently denoted by $M(\ell)$.

The index of discrepancy \eqref{eq:MSE} clearly also depends on $K$, the number of moments used to compute it and $1/K$ in \eqref{eq:MSE} normalizes the indices in order to make them comparable as $K$ varies. A natural question is then about the sensitivity of \eqref{eq:MSE} w.r.t. $K$. It is  desirable for $\ell_M$ to capture fine variations between the theoretical and empirical distributions, which is assured for large $K$. In extensive simulation studies not reported here we noted that increasing $K$ in the range $\{1, \ldots, 10\}$ makes the index increase and then plateau and this holds for all processes and parameter specifications used in the paper. Recalling also the whole body of work by Pearson on eponymous curves, which shows that the knowledge of four moments suffices to cover a large number of known distributions, we adhere to his rule of thumb and choose $K=4$ in our analyses. On the one hand it is a good compromise between targeted precision of the approximation and speed of the algorithm. On the other hand it is straightforward to check the results as $K$ varies in specific applications; for the ones considered in the following sections the differences are negligible.

In the literature several heuristic indices based on the empirical jump sizes around the level of truncation have been discussed \citep[cf Remark 3 in][]{barrios2013modeling}. Here, in order to compare such procedures with our moment criterion, we consider the relative error index which is based on the jumps themselves. It is defined as the expected value of the relative error between two consecutive partial sums of jumps. Its empirical counterpart is denoted by $e_M$ and given by
\begin{equation}\label{eq:heuristic}
e_M = \EFK \bigg[\frac{J_M}{\sum_{i=1}^M J_i}\bigg].
\end{equation}

\section{Applications to Bayesian Nonparametrics} \label{sec:BNP}

In this section we concretely implement the proposed moment-matching \FKa to several Bayesian nonparametric models. The performance in terms of both a priori and a posteriori approximation is evaluated. A comparison of the quality of approximation resulting from using \eqref{eq:heuristic} as benchmark index is provided.

\subsection{A priori simulation study} \label{sec:prior}

We start by investigating the performance of the proposed moment-matching version of the \FKa w.r.t. the CRMs defined in Examples \ref{ex:GG} and \ref{ex:SBP}, namely the generalized gamma and stable-beta processes. Figure~\ref{fig:simulated_examples} displays the behaviour of both the moment-matching distance $\ell_M$ (left panel) and the relative jumps' size index $e_M$  (right panel) as the truncation level $M$ increases. The plots, from top to bottom, correspond to: the \GG with  varying $\gamma$ and $a=1$ fixed; the \IG with varying total mass $a$ (which is a \GG process with $\gamma=0.5$); the \SBP with varying discount parameter $\sigma$ and  $a=1$ fixed.

First consider the behaviour of the indices as the parameter specifications vary. It is apparent that, for any fixed truncation level $M$, the indices $\ell_M$ and $e_M$ increase as each of the parameters $a$, $\gamma$ or $\sigma$ increases.
For instance, roughly speaking, a total mass parameter $a$ corresponds to sampling trajectories defined on the interval $[0,a]$ \citep[see][]{rlp2003}, and a larger interval worsens the quality of approximation for any given truncation level. Also it is natural that $\gamma$ and $\sigma$ impact in similar way $\ell_M$ and $e_M$ given they stand for the ``stable'' part of the L\'evy intensity. See first and third rows of Figure~\ref{fig:simulated_examples}.

As far as the comparison between $\ell_M$ and $e_M$ is concerned, it is important to note that $e_M$  consistently downplays the error of approximation related to the truncation. This can be seen by comparing the two columns of Figure~\ref{fig:simulated_examples}. $\ell_M$ is significantly more conservative than $e_M$ for both the generalized gamma and the stable-beta processes, especially for increasing values of the parameters $\gamma$, $a$ or $\sigma$. This indicates quite a serious issue related to $e_M$ as a measure for the quality of approximation and one should be cautious when using it. In contrast, the moment-matching index $\ell_M$ matches more accurately the known behaviour of these processes as the parameters vary.

By reversing the viewpoint and looking at the truncation level $M(\ell)$ needed for achieving a certain error of approximation $\ell$ in terms of moment-match, the results become even more intuitive. We set $\ell = 0.1$ and computed $M(\ell)$ on a grid of size $20\times 20$ with equally-spaced points for the parameters $(a,\gamma)\in (0,2)\times (0,0.8)$ for the \GG and $(a,c)\in (0,2)\times (0,30)$ for the beta process. Figure~\ref{fig:persp} displays the corresponding plots. In general, it is interesting to note that a limited number of jumps is sufficient to achieve good precision levels. Analogously to Figure~\ref{fig:simulated_examples}, larger values of the parameters require a larger number of jumps to achieve a given precision level. In particular, when $\gamma>0.5$, one needs to sample a significantly larger number of jumps. For instance, in the \GG case, with $a=1$, the required number of jumps increases from $28$ to $53$ when passing from $\gamma=0.5$ to $\gamma=0.75$. It is worth noting that for the normalized version of the \GG, to be discussed in Section \ref{sec:nrmi} and quite popular in applications, the estimated value of $\gamma$ rarely exceeds $0.75$ in species sampling, whereas it is typically in the range $[0.2,0.4]$ in mixture modeling.

\begin{center}
\begin{figure}[ht!]
\includegraphics[width=.49\linewidth]{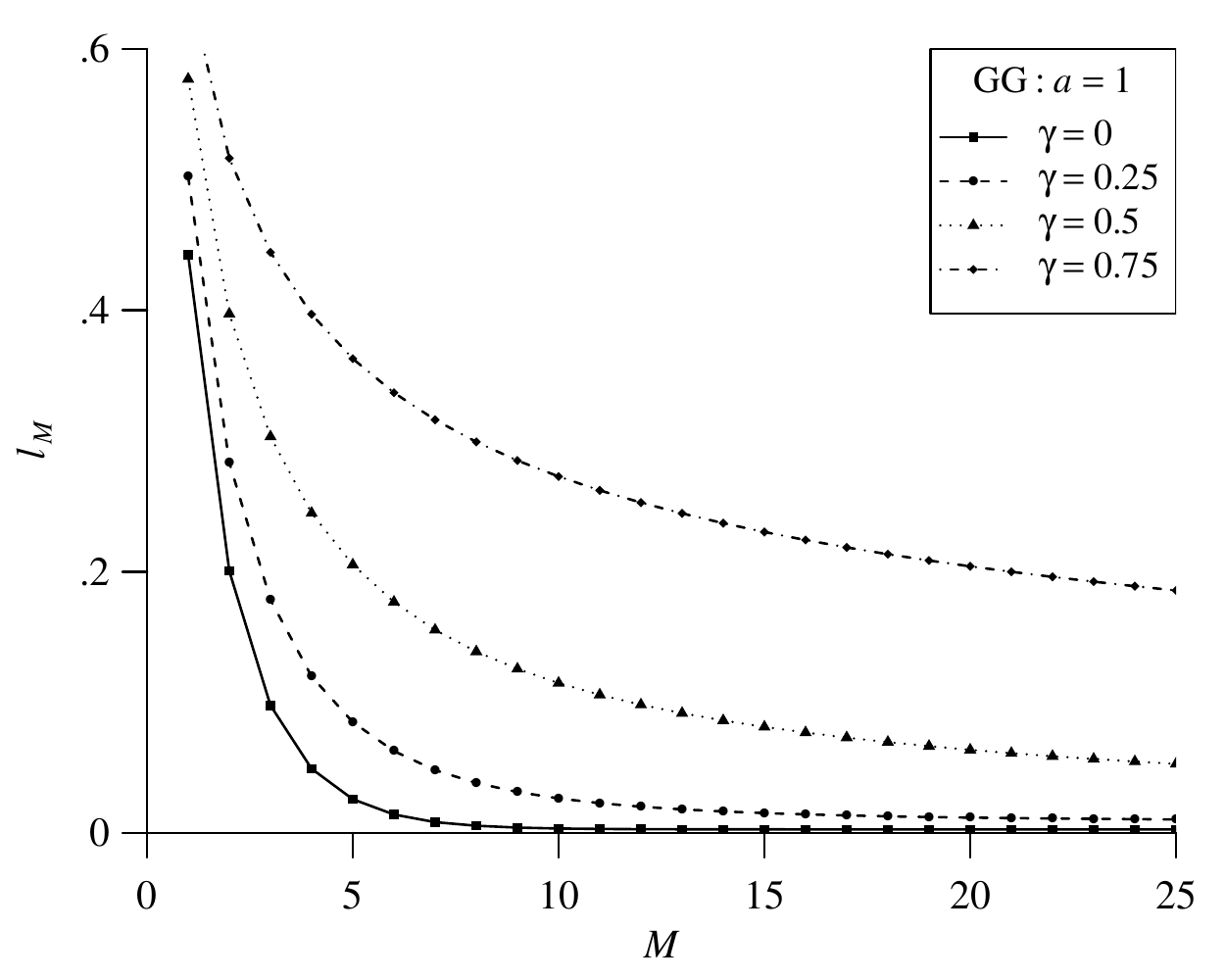}
\includegraphics[width=.49\linewidth]{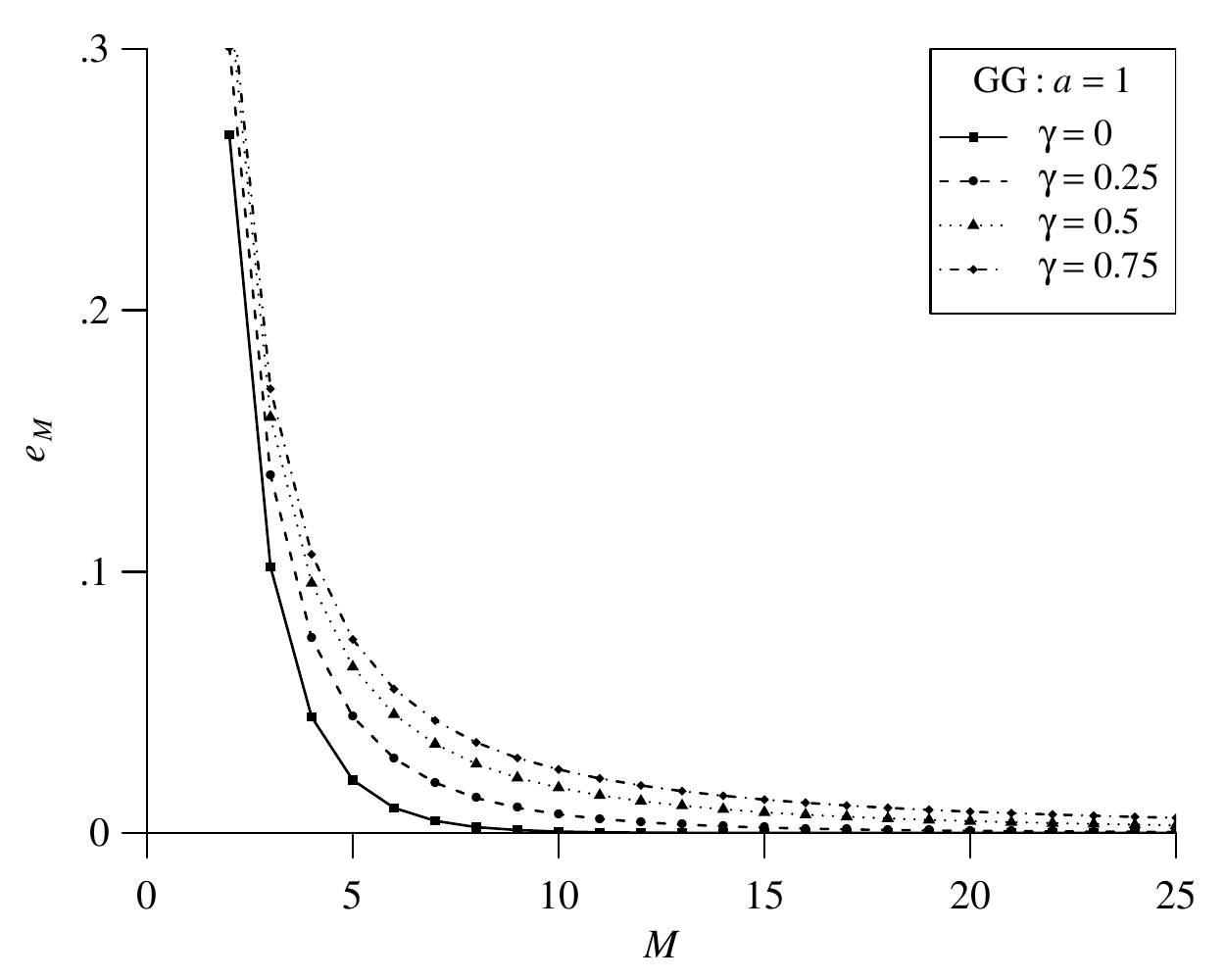}
\includegraphics[width=.49\linewidth]{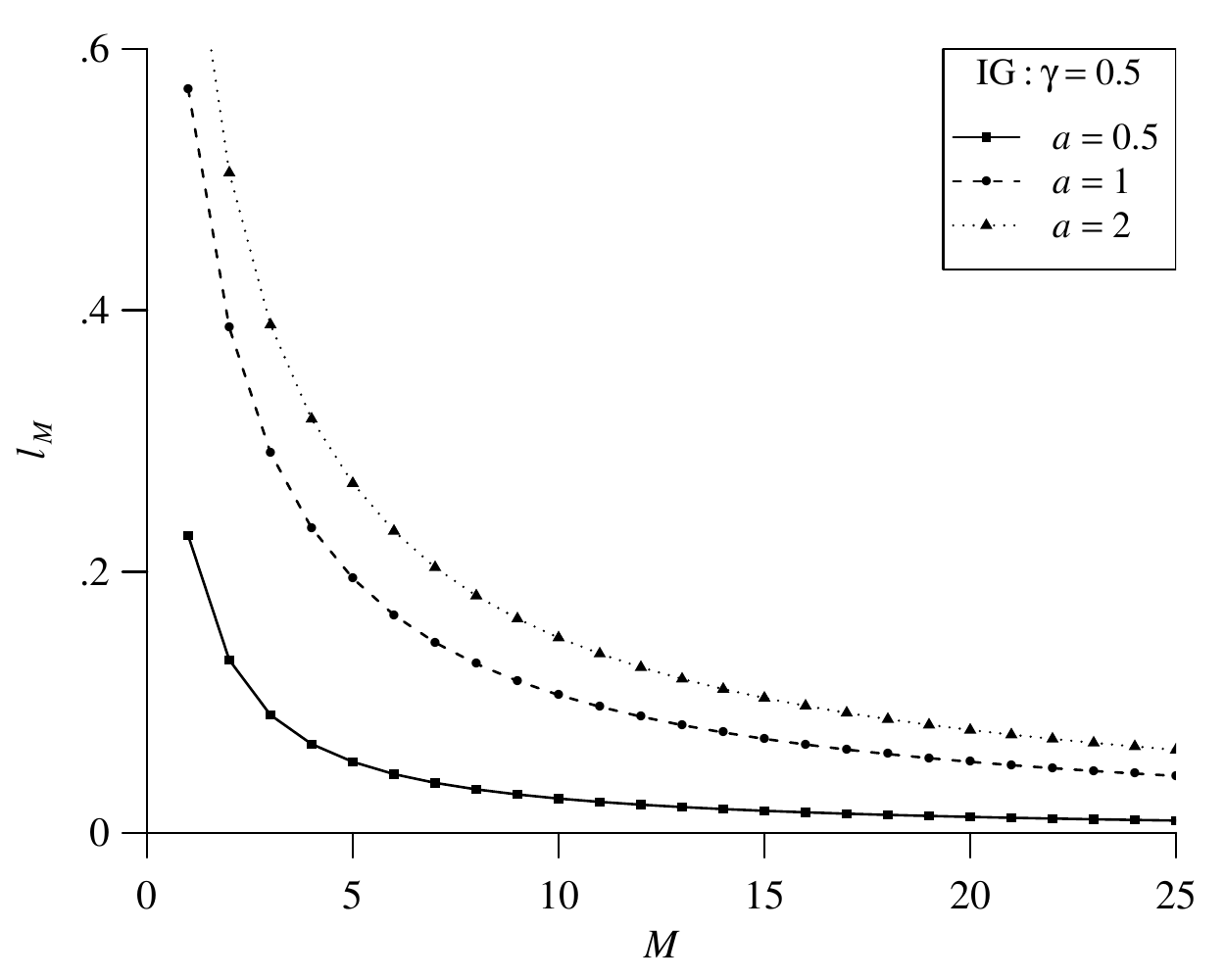}
\includegraphics[width=.49\linewidth]{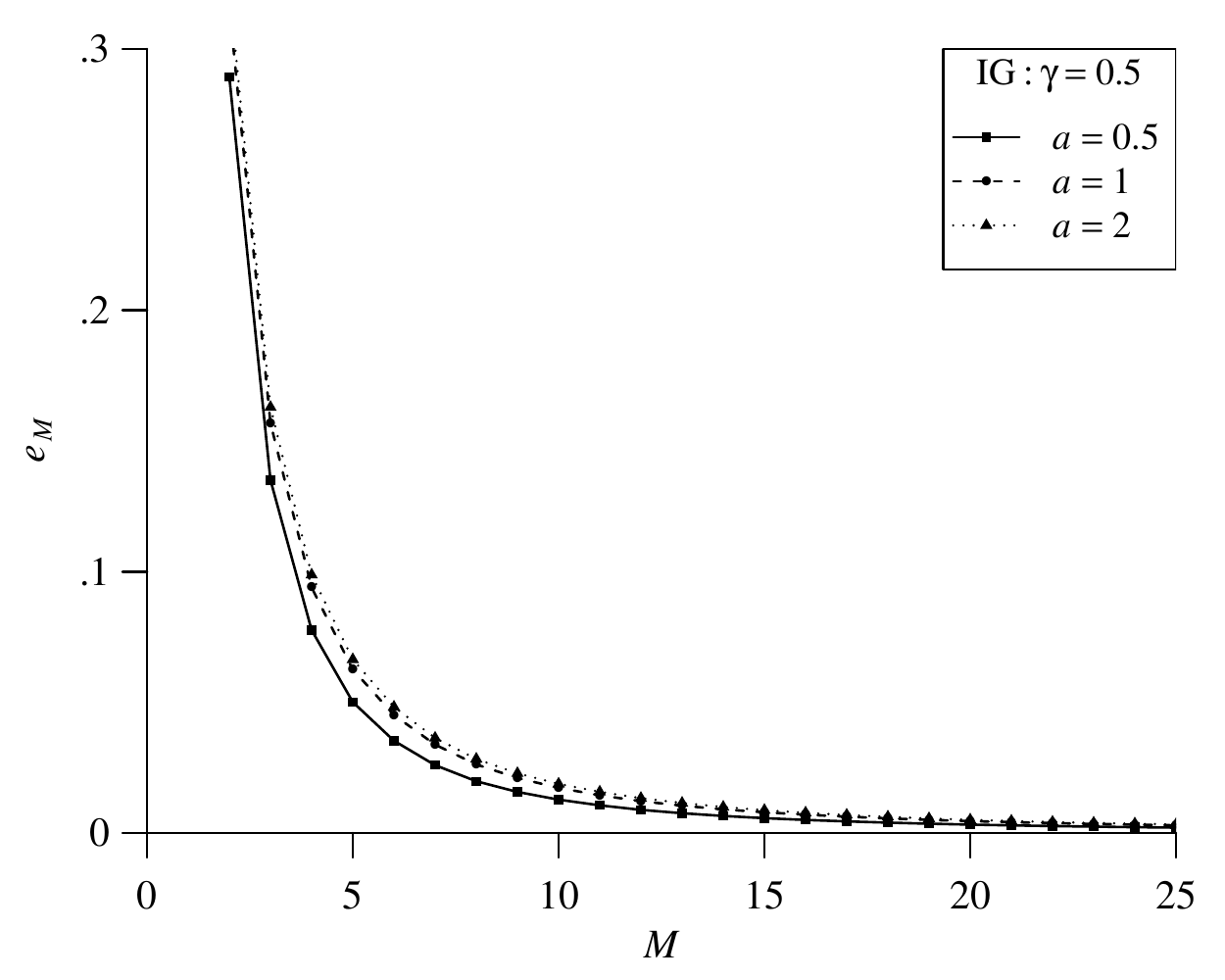}
\subfloat[$\ell_M$]{
\includegraphics[width=.49\linewidth]{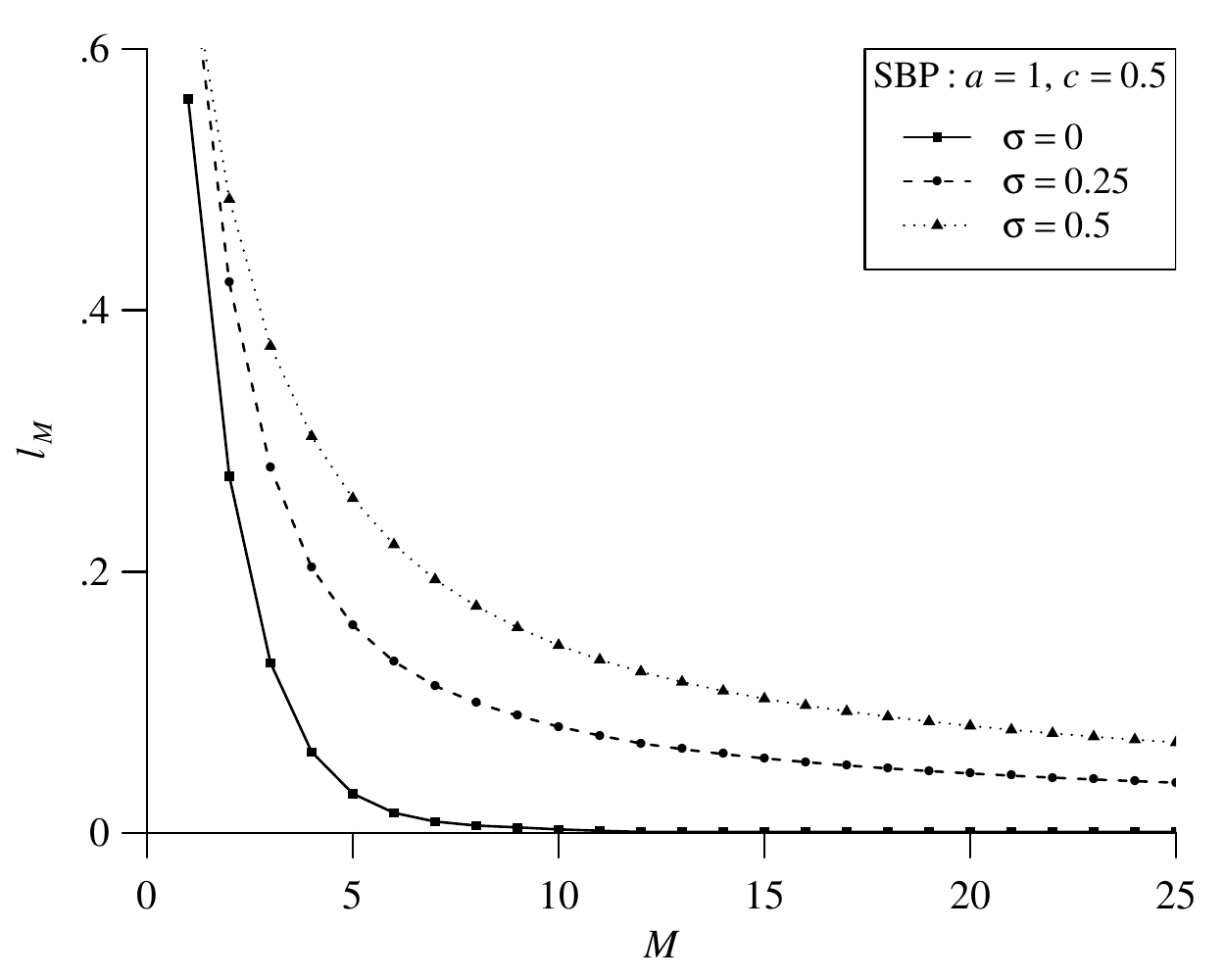}
}
\subfloat[$e_M$]{
\includegraphics[width=.49\linewidth]{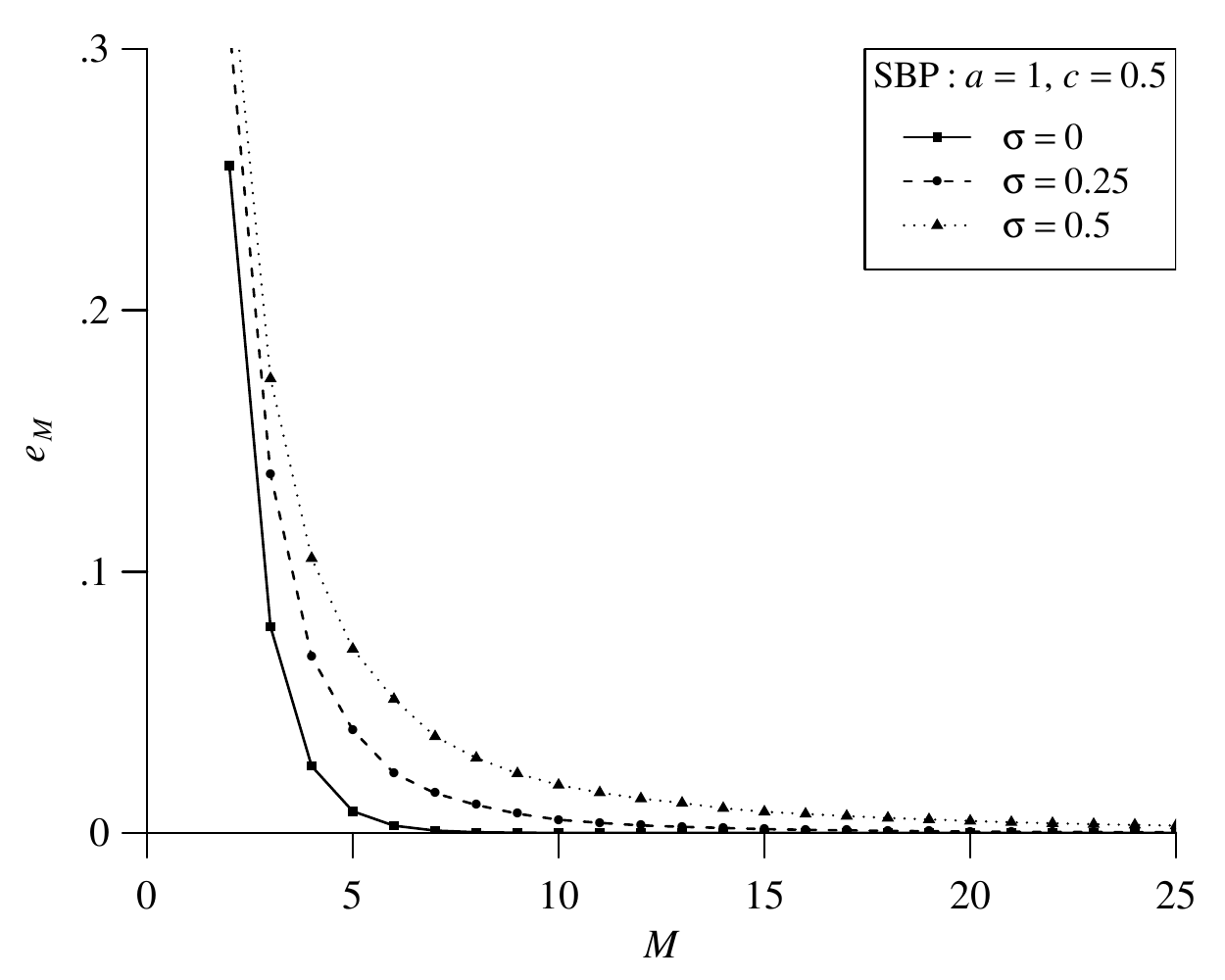}
}
\caption{Left panel: $\ell_M$ as $M$ varies; right panel: $e_M$ as $M$ varies.  Top row: \GG (GG) with  varying $\gamma$ and $a=1$ fixed; middle row: \IG (IG), $\gamma=0.5$, with varying total mass $a$; bottom row: \SBP (SBP) with $a=1$, $c=0.5$ fixed and varying discount parameter $\sigma$. The points are connected by straight lines only for visual simplification.}
\label{fig:simulated_examples}
\end{figure}
\end{center}
\begin{center}
\begin{figure}[ht!]
\subfloat[\GG]{
\includegraphics[width=.49\linewidth]{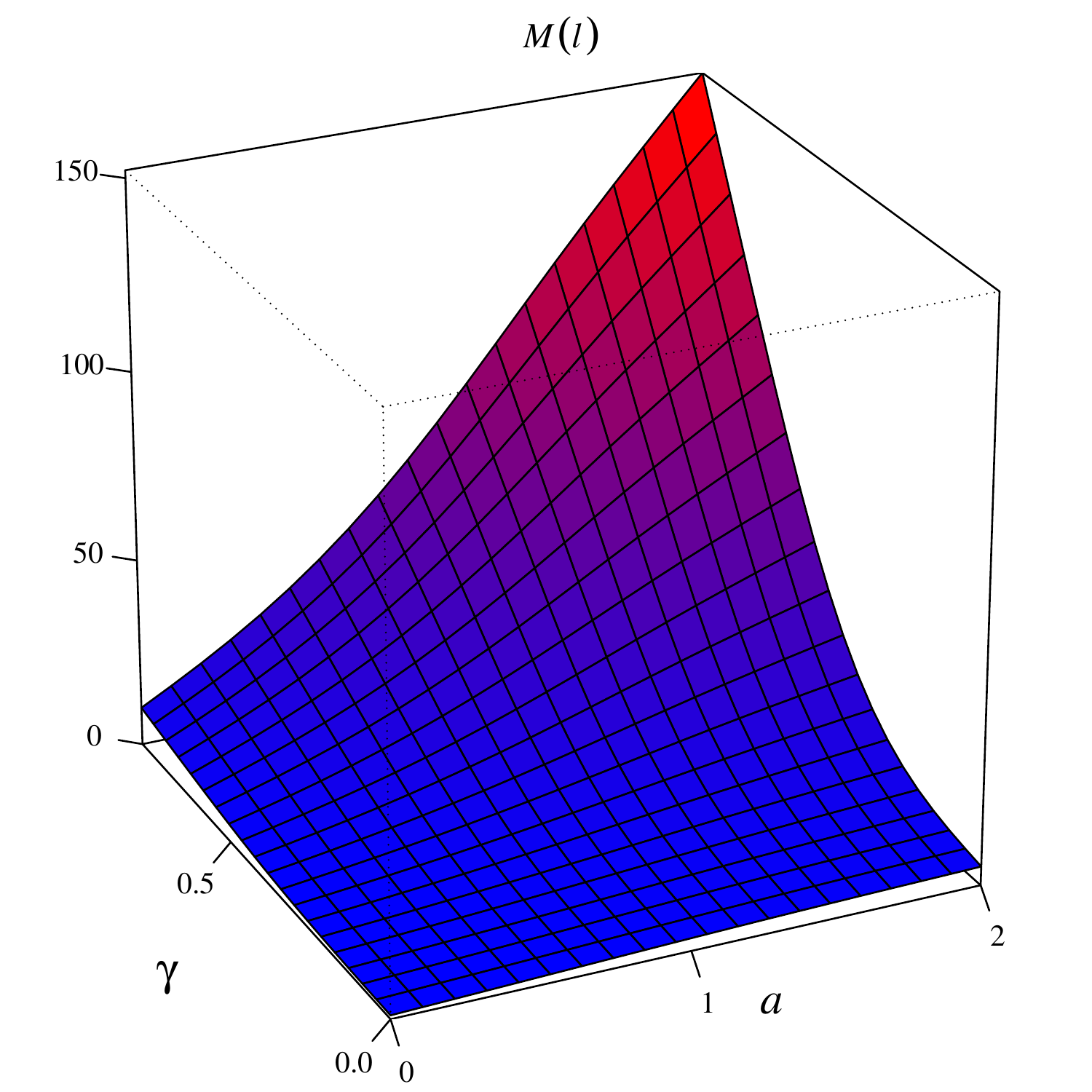}
}
\subfloat[beta process]{
\includegraphics[width=.49\linewidth]{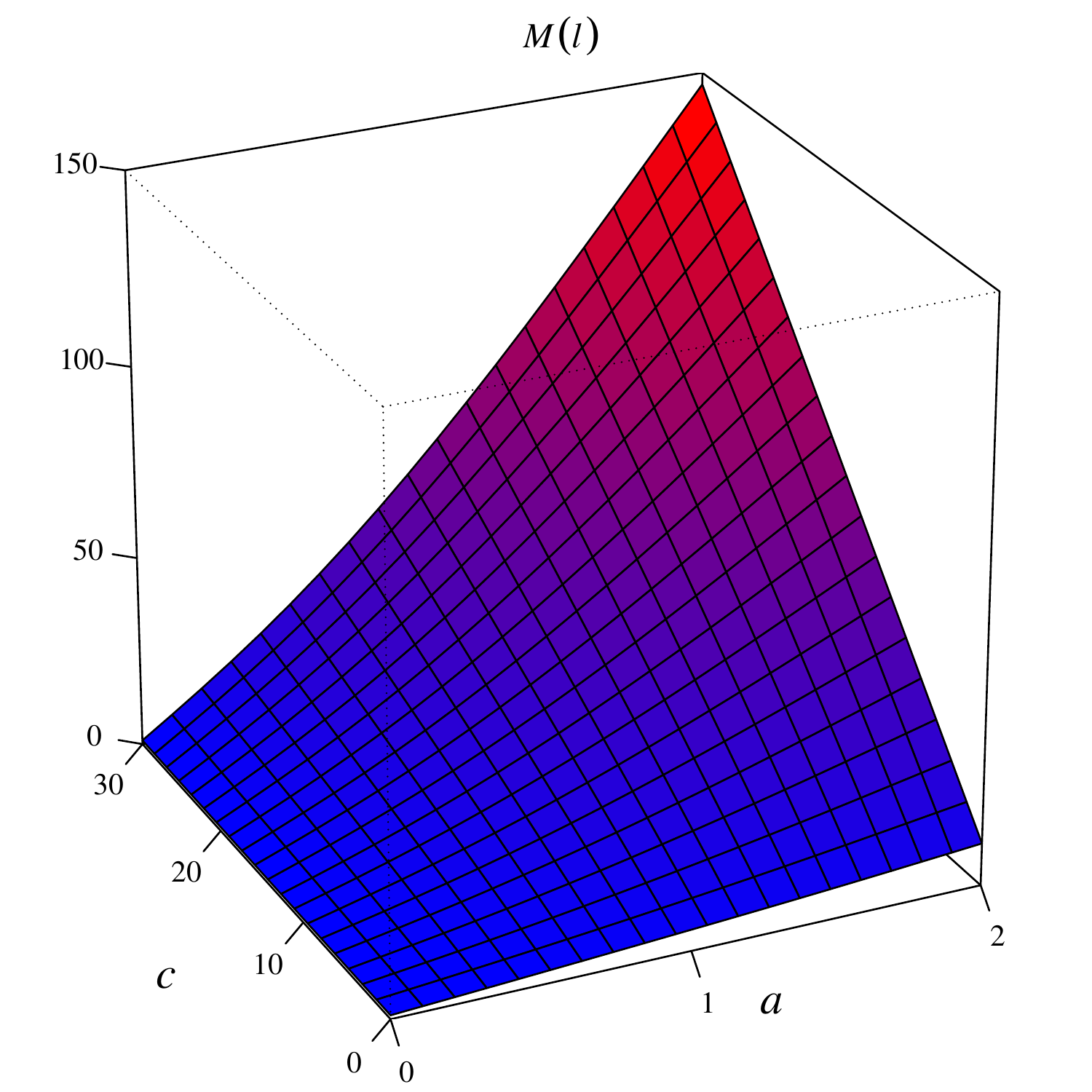}
}
\caption{Number of jumps $M(\ell)$ required to achieve a precision level of $\ell = 0.1$  for $\ell_M$. Left panel: \GG for $a\in(0,2)$ and $\gamma\in(0,0.8)$. Right panel: beta process for $a\in(0,2)$ and $c\in(0,30)$.}
\label{fig:persp}
\end{figure}
\end{center}

\subsection{Normalized random measures with independent increments\label{sec:nrmi}}

Having illustrated the behaviour of the moment-matching methodology for plain CRMs we now investigate it on specific classes of nonparametric priors, which typically involve a transformation of the CRM. Moreover, given their posterior distributions involve updated CRMs it is important to test the moment-matching \FKa also on posterior quantities.  The first class of models we consider are normalized random measures with independent increments (NRMI) introduced by \citet{rlp2003}. Such nonparametric priors have been used as ingredients of a variety of models and in several application contexts. Recent reviews can be found in \citet{lijoi2010models,barrios2013modeling}.

If  $\tilde\mu$ is a CRM with L\'evy intensity \eqref{pim} such that $0<\tilde\mu(\X) < \infty$ (almost surely), then an NRMI is defined as
\begin{equation}
\label{eq:NRMI} \tilde P =\frac{\tilde\mu}{\tilde\mu(\X)}.
\end{equation}
Particular cases of NRMI are then obtained by specifying the CRM in \eqref{eq:NRMI}. For instance, by picking the \GG defined in Example \ref{ex:GG} one obtains the normalized generalied gamma process, denoted by NGG, and first used in a Bayesian context by \citet{lijoi2007controlling}.

\subsubsection{Posterior Distribution of an NRMI}

The basis of any Bayesian inferential procedure is represented by the posterior distribution. In the case of NRMIs, the determination of the posterior distribution is a challenging task since one cannot rely directly on Bayes' theorem (the model is not dominated) and, with the exception of the Dirichlet process, \mbox{NRMIs} are not conjugate as shown in \citet{james2006conjugacy}. Nonetheless, a posterior characterization has been established in \citet{james2009posterior} and it turns out that, even though NRMIs are not conjugate, they still enjoy a sort of ``conditional
conjugacy.'' This means that, conditionally on a suitable latent random variable, the posterior distribution of an NRMI coincides with the distribution of an NRMI having fixed points of discontinuity located at the observations. Such a simple structure suggests that when working with a general NRMI, instead of the Dirichlet process, one faces only one additional layer of difficulty represented by the marginalization with respect to the conditioning latent variable.

Before stating the posterior characterization to be used with our algorithm, we need to introduce some notation and basic facts. Let $(Y_n)_{n\geq 1}$ be an exchangeable sequence directed by an NRMI, i.e.
\begin{align}\label{eq:nonpmic}
&Y_i | \tilde P \simiid \tilde P,\quad \text{ for } i=1,\ldots,n,
\nonumber
\\[-8pt]
\\[-8pt]
\nonumber
&\tilde P \sim Q,
\end{align}
with $Q$ the law of NRMI, and set ${\mathbf Y}=(Y_1, \ldots, Y_n)$. Due to the discreteness of NRMIs, ties will appear with positive probability in ${\mathbf Y}$ and, therefore, the sample information can be encoded by the $K_n=k$ distinct observations $(Y_1^*,\ldots,Y_k^*)$ with frequencies $(n_1, \ldots, n_k)$ such that $\sum_{j=1}^k n_j=n$. Moreover, introduce the nonnegative random variable $U$ such that the
distribution of $[U|\mathbf{Y}]$ has density, w.r.t. the Lebesgue measure,
given by
\begin{equation}
\label{eq:U}\quad f_{U|{\mathbf Y}}(u)\propto u^{n-1}\exp \bigl\{-\psi(u)
\bigr\}\prod_{j=1}^k\tau _{n_j}
\bigl(u|Y_j^* \bigr),
\end{equation}
where $\tau_{n_j}(u|Y_j^*)=\int_0^\infty v^{n_j}\mathrm{e}^{-u v}\rho(\d v|Y_j^*)$ and $\psi$ is the Laplace exponent of $\tilde\mu$ defined by $\psi(u) = -\log\big(L_{\X}(u)\big)$,  cf \eqref{eq:lapl_transform}. Finally, assume $P^*=\mathbb{E}[\tilde P]$  to be nonatomic.

\begin{prop}[\citealp{james2009posterior}]\label{james2}
Let $(Y_n)_{n \geq1}$ be as in \eqref{eq:nonpmic} where $\tilde P$ is an NRMI defined in \eqref{eq:NRMI} with L\'evy intensity as in
\eqref{pim}. Then the posterior distribution of the unnormalized CRM $\tilde \mu$, given a sample ${\mathbf Y}$, is a mixture of the distribution of
$[\tilde\mu|U,\mathbf{Y}]$ with respect to the distribution of $[U|\mathbf{Y}]$. The latter is identified by \eqref{eq:U}, whereas $[\tilde\mu|U=u,\mathbf{Y}]$ is equal in distribution to a CRM with fixed points of discontinuity at the distinct observations $Y_j^*$,
\begin{equation}
\label{eq:post_CRM} \tilde\mu^*+\sum_{j=1}^k
J_j^*\delta_{Y_j^*}
\end{equation}
such that:
\begin{itemize}
\item[(a)] $\tms$ is a CRM characterized by the L\'evy intensity
\begin{equation}
\label{eq:post_levy} \nu^*(\d v, \d x)=\mathrm{e}^{-u v}\nu(\d v, \d x),
\end{equation}
\item[(b)]
the jump height $J_j^*$ corresponding to $Y_j^*$ has density, w.r.t.
the Lebesgue measure, given by
\begin{equation}
\label{eq:post_jump} f_{j}^*(v)\propto
v^{n_j}\mathrm{e}^{-u v}\rho \bigl(\d v|Y_j^* \bigr),
\end{equation}
\item[(c)] $\tms$ and $J_j^*$, $j=1,\ldots,k$, are
independent.
\end{itemize}

Moreover, the posterior distribution of the NRMI $\tilde P$,
conditional on $U$, is given by
\begin{equation}
\label{eq:post.ncrm}\quad [\tilde P | U, \mathbf{Y}] \stackrel{d} {=} w
\frac{\tilde\mu^*}{\tilde\mu^*(\X)}+(1-w) \frac{\sum_{k=1}^k
J_j^*
\delta_{Y_j^*}}{\sum_{l=1}^k J_l^*},
\end{equation}
where $w=\tilde\mu^*(\X)/(\tilde\mu^*(\X)+\sum_{l=1}^k J_l^{*})$.
\end{prop}

In order to simplify the notation, in the statement we have omitted explicit reference to the dependence on $[U|\mathbf{Y}]$ of both $\tilde \mu^*$ and $\{J_j^*\dvtx j=1,\ldots,k\}$, which is apparent from \eqref{eq:post_levy} and~\eqref{eq:post_jump}.
A nice feature of the posterior representation of Proposition~\ref{james2} is that the only quantity needed for deriving explicit expressions for particular cases of NRMI is the L\'evy intensity \eqref{pim}. For instance, in the case of the \GG, the CRM part $\tms$ in \eqref{eq:post_CRM} is still a \GG characterized by a L\'evy intensity of the form of~\eqref{eq:ngg}
\begin{equation}
\label{eq:post_levy_ngg} \nu^*(\d v,\d y)= \frac{\edr^{-(1+u)v}}{\Gamma(1-\gamma) v^{1+\gamma}}\,\d v\, a P^*(\d y).
\end{equation}

Moreover, the distribution of the jumps \eqref{eq:post_jump} corresponding to the fixed points of discontinuity $Y_j^*$'s in \eqref{eq:post_CRM} reduces to a gamma distribution with density
\begin{equation}
\label{postfj} f_{j}^*(v)=\frac{(1+u)^{n_j-\gamma}}{\Gamma(n_j-\gamma)} v^{n_j-\gamma-1}
\mathrm{e}^{-(1+u)v}. 
\end{equation}

Finally, the conditional distribution of the non-negative latent variable $U$ given ${\mathbf Y}$ \eqref{eq:U} is given by
\begin{equation}
\label{condu} f_{U|{\mathbf Y}}(u)\propto u^{n-1}(u+
1)^{k\gamma-n}\exp \biggl\{ -\frac
{a}{\gamma}(u+1)^{\gamma}
\biggr\}.
\end{equation}

The availability of this posterior characterization makes it then possible to determine several important quantities such as the
predictive distributions and the induced partition distribution. See \citet{james2009posterior} for general NRMI and \citet{lijoi2007controlling} for the subclass of normalized generalized gamma processes. See also~\citet{argiento2015blocked} for another approach to approximate the NGG with a finite number of jumps.

\subsubsection{Moment-matching for posterior NRMI}

From \eqref{eq:post_CRM} it is apparent that the posterior of the unnormalized CRM $\tilde \mu$, conditional on the latent variable $U$, is composed of the independent sum of a CRM $\tilde \mu^*$ and fixed points of discontinuity at the distinct observations $Y_j^*$. The part which is at stake here is obviously  $\tilde \mu^*$ for which only approximate sampling is possible. As for the fixed points of discontinuities, they are independent from $\tilde \mu^*$ and can be sampled exactly, at least in special cases.

We focus on the case of the NGG process. By~\eqref{eq:post_levy} the L\'evy intensity of $\tilde \mu^*$ is obtained by exponentially tilting the L\'evy intensity of the prior $\tilde \mu$. Hence, the \FKa applies in the same way as for the prior. The sampling of the fixed points jumps is straightforward from the gamma distributions~\eqref{postfj}.
As far as the moments are concerned, key ingredient of our algorithm, the cumulants of $\tilde \mu^*$ are equal to $\kappa_i^*= a\frac{{(1-\gamma)_{(i-1)}}}{(u+1)^{i-\gamma}}$ and the corresponding moments are then obtained via Proposition \ref{prop:crm_moments}.

Our simulation study is based on a sample of size $n=10$. Such a small sample size is challenging in the sense that the data provide rather few information and the CRM part of the model is still prevalent. We examine three possible clustering configurations of the observations $Y_i^*$s: (i) $k=1$ group, with $n_1=10$, (ii) $k=3$ groups, with $n_1=1,\,n_2=3,\,n_3=6$, and (iii) $k=10$ groups, with $n_j=1$ for $j=1,\ldots,10$. First let us consider the behaviour of $f_{U|{\mathbf Y}}$, which is illustrated in Figure~\ref{fig:U-density} for $n=10$ and $k\in\{1,2,\ldots,10\}$. It is clear that the smaller the number of clusters, the more $f_{U|{\mathbf Y}}$ is concentrated on small values, and vice versa.

\begin{center}
\begin{figure}[ht!]
\includegraphics[width=.49\linewidth]{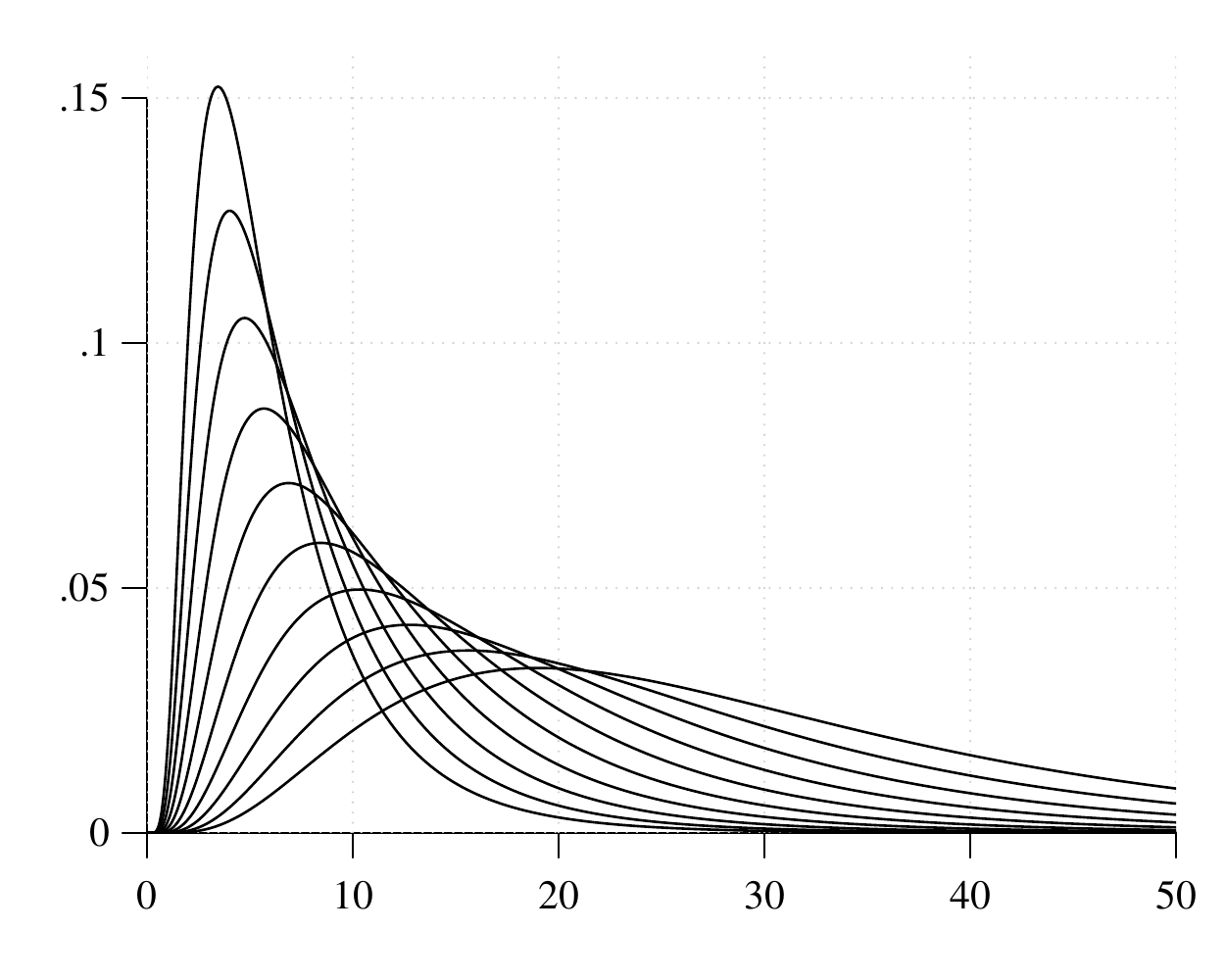}
\caption{NGG posterior: density $f_{U|{\mathbf Y}}$ with $n=10$ observations, $a=1$, $\gamma=0.5$, and number of clusters $k \in \{1, \ldots, 10\}$; $k=1$ corresponds to the most peaked density and $k=10$ to the flattest.}
\label{fig:U-density}
\end{figure}
\end{center}

Now we consider $\tilde \mu^*(\X)$, the random total mass corresponding to the CRM part of the posterior only given in \eqref{eq:post.ncrm}. Such a quantity depends on $U$ whose distribution is driven by the data  $\mathbf Y$. In order to keep the presentation as neat as possible, and in the same time to remain consistent with the data, we choose to condition on $U=u$ for $u$ equal to the mean of $f_{U|{\mathbf Y}}$, the most natural representative value. Given this, it is possible to run the \FKa on the CRM part $\tilde \mu^*$ of the posterior and compute moment-matching index $\ell_M$ as the number of jumps varies. Figure~\ref{fig:NRMI} shows these results for the \IGau CRM, a special case of the \GG corresponding to $\gamma=0.5$. Such posteriors were sampled under the above mentioned $\mathbf Y$ clustering configuration scenarios (i)-(iii), which led to mean values of $U|{\mathbf Y}$ of, respectively, $6.3$, $8.9$ and $25.1$. The plot also  displays a comparison to the prior values of $\ell_M$ and indicates that for a given number of jumps the approximation error, measured in terms of $\ell_M$, is smaller for the posterior CRM part $\tilde \mu^*$ w.r.t. to the prior CRM $\tilde \mu$.

Additionally, instead of considering only the CRM part $\tilde \mu^*$ of the posterior, one may be interested in the quality of the full posterior which includes also the fixed discontinuities. For this purpose we consider an index which is actually of interest in its own. In particular, we evaluate the relative importance of the CRM part w.r.t. the part corresponding to the fixed points of discontinuity in terms of the ratio $\E\big(\sum_{j=1}^k J_j^*\big)/\E\big(\tilde\mu^*(\X)\big)$. Loosely speaking one can think of the numerator as the expected weight of the data and the denominator as the expected weight of the prior. Recall that in the NGG case, for a given pair $(n,k)$ and conditional on $U=u$, the sum of fixed location jumps is a $\text{gamma}(n-k\gamma,u+1)$. Hence, the index becomes
    \begin{equation}\label{eq:relative}
    \frac{\E\big(\sum_{j=1}^k J_j^*\vert U=u\big)}{\E\big(\tilde\mu^*(\X)\vert U=u\big)} = \frac{(n-k\gamma)/(u+1)}{a/(u+1)^{1-\gamma}}= \frac{n-k\gamma}{a(u+1)^{\gamma}}.
    \end{equation}

By separately mixing the conditional expected values in \eqref{eq:relative} over $f_{U|{\mathbf Y}}$ (we use an adaptive rejection algorithm to sample from $f_{U|{\mathbf Y}}$) we obtained the results summarized in the table of Figure~\ref{fig:NRMI}.
We can appreciate that the fixed part typically overcomes (or is at least of the same order than) the CRM part, a phenomenon which uniformly accentuates as the sample size $n$ increases. Returning to the original problem of measuring the quality of approximation in terms of moment matching, these findings make it apparent that the comparative results of Figure~\ref{fig:NRMI} between prior and posterior are conservative. In fact, if performing the moment-match on the whole posterior, i.e. including the fixed jumps which can be sampled exactly, the corresponding moment-matching index would, for any given truncation level $M$, indicate a better quality of approximation w.r.t. the index based solely on $\tilde \mu^*$. Note that computing the moments of $\tilde \mu^*(\X)+\sum_{i=1}^k J_i$ straightforward given the independence between $\tilde \mu^*$ and the fixed jumps $J_i$'s and also among the jumps themselves. From a practical point of view the findings of this section suggest that a given quality of approximation $\ell$ in terms of moment-match for the prior represents an upper bound for the quality of approximation in the posterior.

\begin{center}
\begin{figure}[ht!]
\begin{minipage}{7cm}
\includegraphics[width=\linewidth]{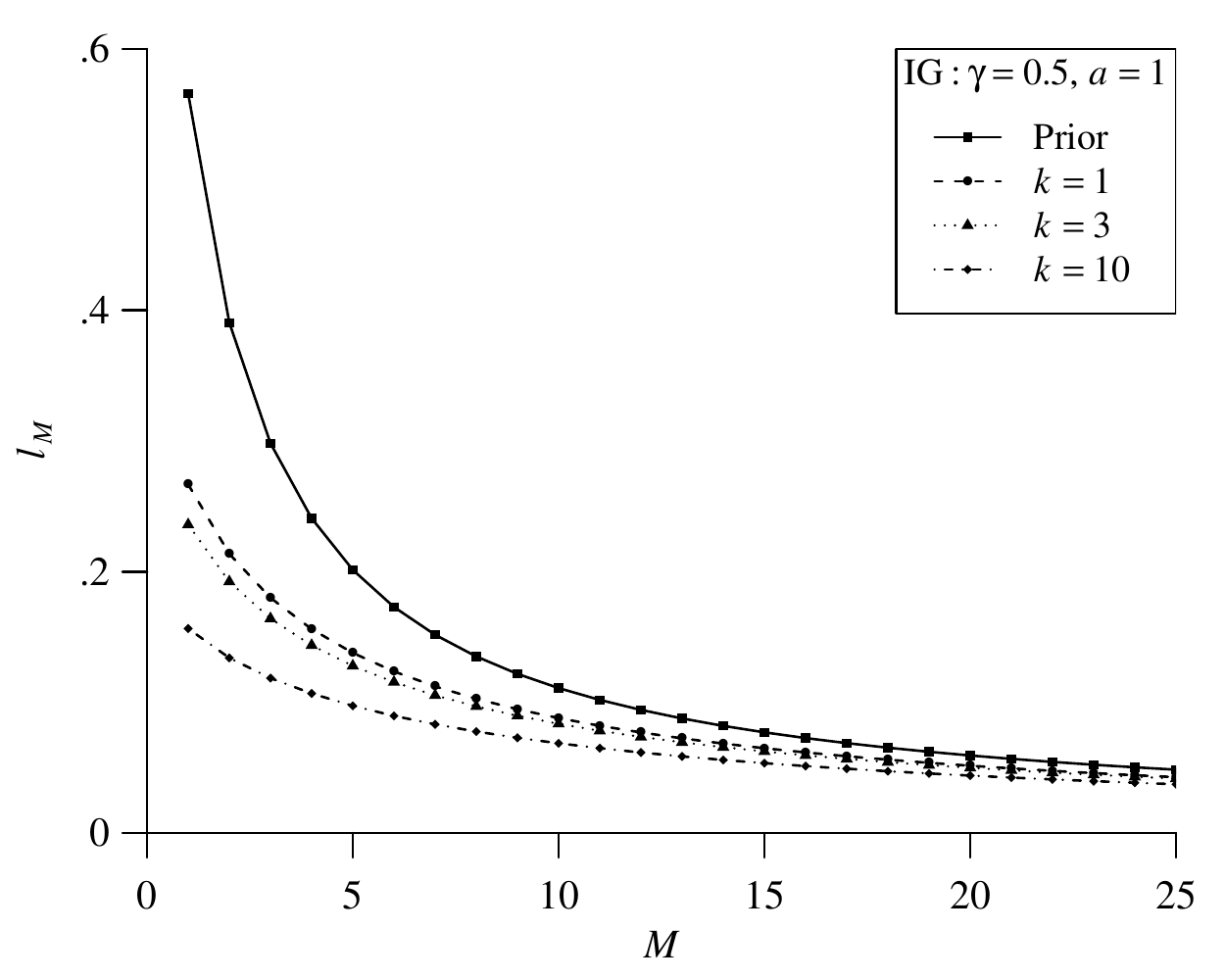}
\end{minipage}
\hskip1cm
\begin{minipage}{5cm}
\begin{tabular}{cccc}
\toprule
& \multicolumn{3}{ c }{$\E\big(\sum_{j=1}^k J_j^*\big)/\E\big(\tilde\mu^*(\X)\big)$} \\
  \hline
$k\quad \backslash \quad n$& 10&30&100\\
  \hline
$1$ & 3.34 & 7.30 & 13.50 \\
$\sqrt{n}$ &  2.65 & 4.68 & 6.05 \\
$n$ &  0.89 & 0.98 & 0.99 \\
\bottomrule
\end{tabular}
\end{minipage}
\caption{Inverse-Gaussian process ($\gamma = 0.5$) with $a=1$. Left: Moment-matching errors $\ell_M$ as the number of jumps $M$ varies. $\ell_M$ corresponding to prior $\tilde \mu$  (continuous line) and posterior $\tilde \mu^*$ under $\mathbf{Y}$ clustering scenarios (i) (dashed line), (ii) (dotted line), (iii) (dotted-dashed line).
Right: Index of relative importance $\E\big(\sum_{j=1}^k J_j^*\big)/\E\big(\tilde\mu^*(\X)\big)$ 
for varying $(n,k)$.}
\label{fig:NRMI}
\end{figure}
\end{center}

\subsubsection{A note on the inconsistency for diffuse distributions\label{sec:asymp}}

In the context of Gibbs-type priors, of which the normalized \GG  is a special case, \citet{de2012asymptotic} showed that, if the data are generated from a ``true'' $P_0$, the posterior of $\tilde P$  concentrates at a point mass which is the linear combination
\begin{align*}
b P^*(\cdot) +(1-b)P_0(\cdot)
\end{align*}
of the prior guess $P^* = \E(\tilde P)$ and $P_0$. The weight $b$ depends on the prior and, indirectly, on $P_0$, since $P_0$ dictates the rate at which the distinct observations $k$ are generated. For a diffuse $P_0$, all observations are distinct and $k=n$ (almost surely). In the NGG case this implies that $b=\gamma$ and hence the posterior is inconsistent since it does not converge to $P_0$. For the inverse-Gaussian process, i.e. with $\gamma=0.5$, the posterior distribution gives asymptotically the same weight to $P^*$ and $P_0$.
The last row of the table of Figure~\ref{fig:NRMI}, which displays the ratio $\E\big(\sum_{j=1}^k J_j^*\big)/\E\big(\tilde\mu^*(\X)\big)$ for $k=n$, is an illustration of this inconsistency result since the ratio gets close to $1$ as $n$ grows. In contrast, when $P_0$ is discrete, which implies that $k$ increases at a slower rate than $n$, one always has consistency. This is illustrated by the first two rows of the table of Figure~\ref{fig:NRMI}, where one can appreciate that the ratio $\E\big(\sum_{j=1}^k J_j^*\big)/\E\big(\tilde\mu^*(\X)\big)$ increases as $n$ increases, giving more and more weight to the data. These findings suggest that consistency issues for general NRMI could be explored from new perspectives based on the study of the asymptotic behavior of $f_{U|{\mathbf Y}}$, which will be subject to future work.

\subsection{Stable-beta Indian buffet process\label{sec:ibp}}

The Indian buffet process (IBP), introduced in \citet{ghahramani2005infinite}, is one of the most popular models for feature allocation and is closely connected to the beta process discussed in Example \ref{ex:SBP}. In fact, when marginalizing out the Dirichlet process and considering the resulting partition distribution one obtains the well known Chinese restaurant process. Likewise, as shown in \citet{thibaux2007hierarchical}, when integrating out a beta process in a Bernoulli process (\BeP) model one obtains the IBP. Recall that a Bernoulli process, with an atomic base measure $\tilde\mu$, is a stochastic process whose realizations are collections of atoms of mass 1, with possible locations given by the atoms of the  base measure $\tilde\mu$. Such an atom is element of the collection with probability given by the jump size in $\tilde\mu$. Later, \citet{teh2009indian} generalized the construction and defined the stable-beta \IBP as
\begin{align}\label{eq:BeP-model}
&Y_i \vert \tilde \mu \simiid \BeP(\tilde \mu)\quad \text{for } i = 1,\ldots,n,
\nonumber
\\[-8pt]
\\[-8pt]
\nonumber
&\tilde \mu \vert c, \sigma, aP^* \sim \text{SBP}(c,\sigma,aP^*).
\end{align}

Given the construction involves a CRM, it is clear that any conditional simulation algorithm will need to rely on some truncation for which we use our moment-matching \FKa.

\subsubsection{Posterior distribution in the IBP}
Let us consider a conditional iid sample ${\mathbf Y}=(Y_1, \ldots, Y_n)$ as in \eqref{eq:BeP-model}. Note that due to the discreteness of $\tilde \mu$, ties appear with positive probability. We adopt the same notations for the ties $Y_j^*$ and frequencies $n_j$ as in Section~\ref{sec:nrmi}. Then we can state the following result which highlights the  posterior structure of the \SBP in the \IBP.
\begin{prop}[\citealp{teh2009indian}]\label{thm:hjort}
Let $(Y_n)_{n \geq1}$ be as in \eqref{eq:BeP-model}. Then the posterior distribution of $\tilde \mu$ conditional on ${\mathbf Y}$ is given by the distribution of
\begin{equation*}
\tilde\mu^*+\sum_{j=1}^k
J_j^*\delta_{Y_j^*}
\end{equation*}
where
\begin{itemize}
\item[(a)] $\tms$ is a \SBP characterized by the L\'evy intensity
%
\begin{equation*}
\nu^*(\d v, \d x)=(1-v)^{n}\nu(\d v, \d x),
\end{equation*}
\item[(b)]
the jump height $J_j^*$ corresponding to $Y_j^*$ is beta distributed
%
\begin{equation*}
J_{j}^*\sim \emph{beta}(n_j-\sigma,c+\sigma+n- n_j),
\end{equation*}
\item[(c)] $\tms$ and $J_j^*$, $j=1,\ldots,k$, are
independent.
\end{itemize}
\end{prop}

Note that due to the polynomial tilting of $\nu$ by $(1-u)^n$ in \textit{(a)} above, the CRM part $\tilde\mu^*$ is still a \SBP with updated parameters
$$c^* = c+n \text{ and } a^* = a\frac{(c+\sigma)_{(n)}}{(c+1)_{(n)}},$$
while the discount parameter $\sigma$ remains unchanged.

\subsubsection{Moment-matching for the IBP}
In order to implement the moment-matching methodology we first need to evaluate the posterior moments of the random total mass. For this purpose, we rely on the moments characterization in terms of the cumulants provided in Proposition~\ref{prop:crm_moments}. The cumulants $\kappa_i^*$ of the CRM part $\tms(\X)$ are obtained from Table~\ref{tab:moments} with the appropriate parameter updates which leads to
$$\kappa_i^*=a^*\frac{(1-\sigma)_{(i-1)}}{(1+c^*)_{(i-1)}}=a\frac{(1-\sigma)_{(i-1)}(c+\sigma)_{(n)}}{(1+c)_{(n+i-1)}}.$$

We consider two stable-beta processes: the beta process prior $\tilde\mu\sim\text{SBP}(c=1,\sigma=0,a=1)$ and the stable-beta process prior $\tilde\mu\sim\text{SBP}(c=1,\sigma=0.5,a=1)$. We let $n$ vary in $\{5,10,20\}$. In contrast to the NRMI case, there is no need to work under different scenarios for the clustering profile of the data, since the posterior CRM $\tilde \mu^*$ is not affected by them with only the sample size entering the updating scheme. We  compare the prior moment-match for $\tilde \mu$ with the posterior moment-match for $\tilde \mu^*$ in terms of our discrepancy index $\ell_M$ and the results are displayed in Figure~\ref{fig:IBP}. The comparison shows that there is a gain in precision between prior and posterior distributions in terms of $\ell_M$ suggesting that the a priori error level $\ell$ represents an upper bound for the posterior approximation error.

As in Section~\ref{sec:nrmi}, we also evaluate the relative weights of fixed jumps and posterior CRM or, roughly, of the data w.r.t. the prior. Recalling that fixed location jumps $J_j^*$ are independent and $\text{beta}(n_j-\sigma,c+\sigma+n- n_j)$ and some algebra allow to re-write the ratio of interest as
    $$
    \frac{\E\big(\sum_{j=1}^k J_j^*\big)}{\E\big(\tilde\mu^*(\X)\big)} =
    \frac{(n-k\sigma)(c+1)_{(n-1)}}{a(c+\sigma)_{(n)}}.
    $$

Table~\ref{tab:SBP_ratio} displays the corresponding values for different choices of $n$ and $k$. As in the NRMI case, the fixed part overcomes the CRM part, which means that the data dominate the prior, and, moreover, their relative weight increases as $n$ increases. In terms of moment-matching this shows that, if one looks at the overall posterior structure, the approximation error connected to the truncation is further dampened.

\begin{center}
\begin{figure}[ht!]
\subfloat[beta process ($\sigma=0$)]{
\includegraphics[width=.45\linewidth]{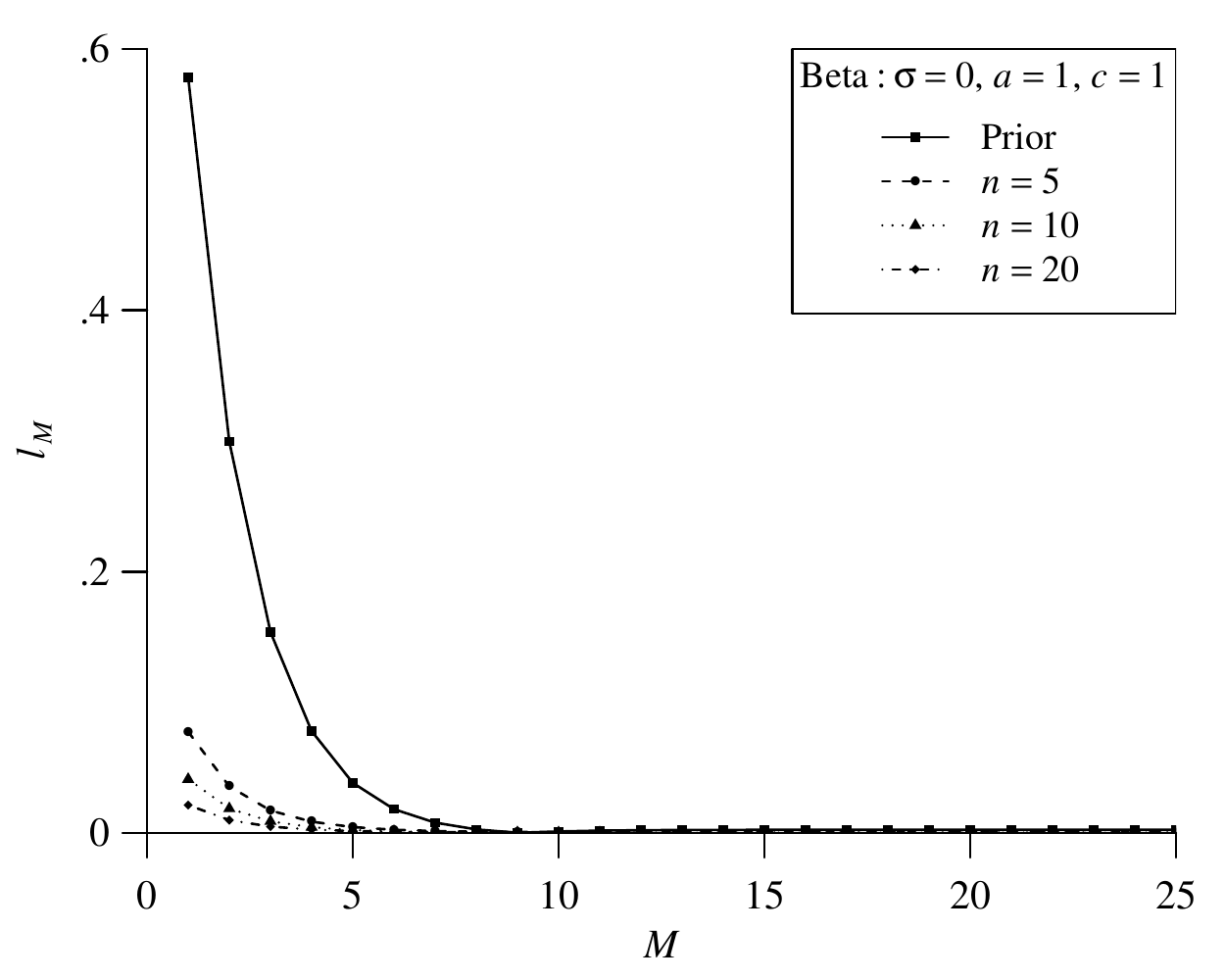}
}
\subfloat[stable beta process ($\sigma=0.5$)]{
\includegraphics[width=.45\linewidth]{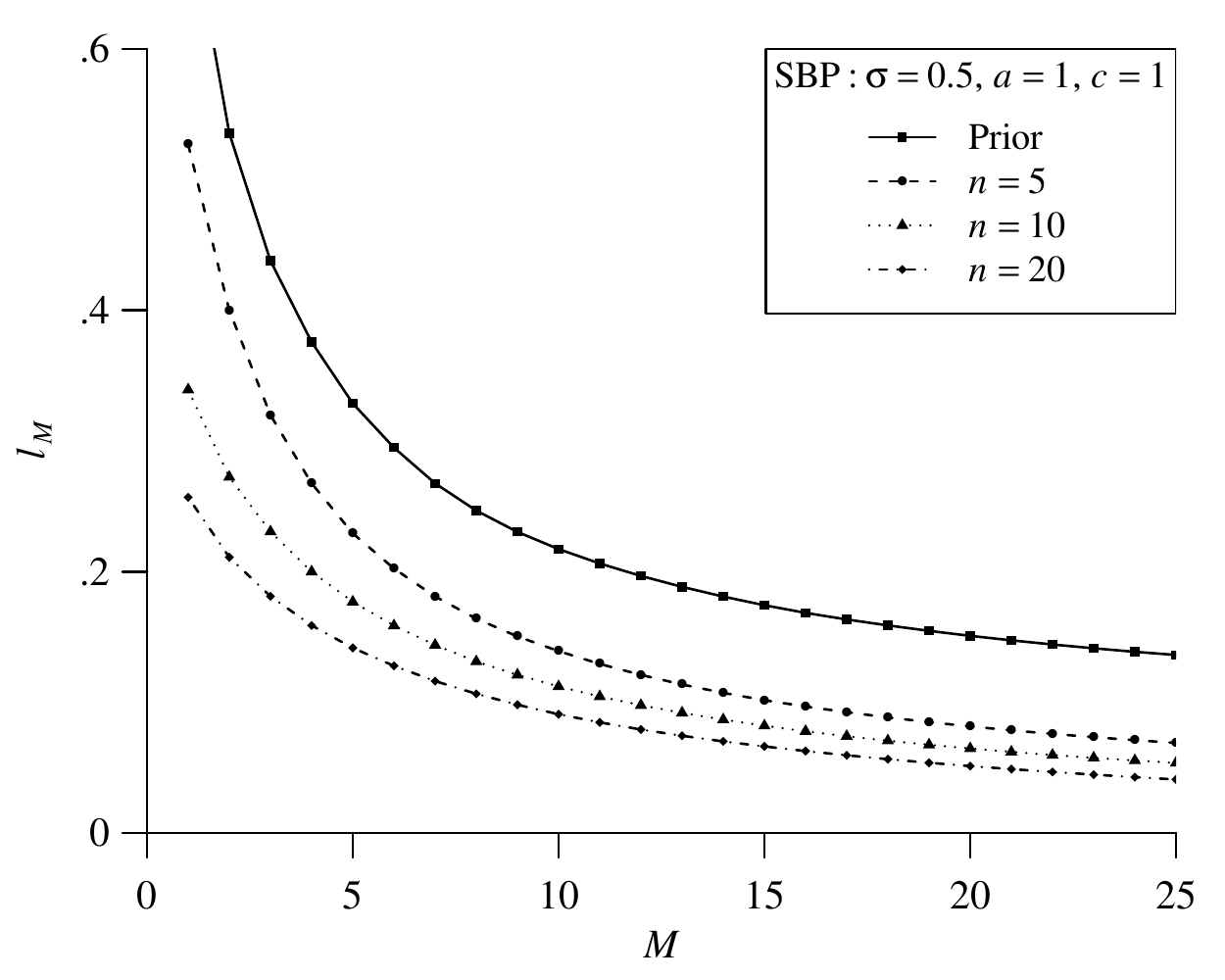}
}
\caption{Moment-matching errors $\ell_M$ as the number of jumps $M$ varies for the \SBP with $c=1$, $a=1$, and, respectively, $\sigma=0$ (left panel) and $\sigma=0.5$ (right panel). $\ell_M$ corresponding to prior $\tilde \mu$  (continuous line) and the posterior $\tilde \mu^*$ given with $n=5$ (dashed line) and $n=10$ (dotted line)  and $n=20$ (dashed-dotted line) observations.}\label{fig:IBP}
\end{figure}
\end{center}
\begin{center}
\begin{table}[ht!]
\begin{tabular}{cccc}
\toprule
&
\multicolumn{3}{ c }{$\E\big(\sum_{j=1}^k J_j^*\big)/\E\big(\mu^*(\X)\big)$} \\
\hline
$k\quad \backslash \quad n$& 10&30&100\\
  \hline
$1$ & 2.57 & 4.71 & 8.79 \\
$n^\sigma$ &    2.28 & 4.36 & 8.39 \\
$n$ &    1.35 & 2.40 & 4.41 \\
\bottomrule
\end{tabular}
\caption{\label{tab:SBP_ratio} Stable-beta process with $\sigma=0.5$, $c=1$ and $a=1$: Index of relative importance $\E\big(\sum_{j=1}^k J_j^*\big)/\E\big(\mu^*(\X)\big)$ for varying $(n,k)$.}
\end{table}
\end{center}

\subsection{Practical use of the moment-matching criterion\label{sec:practical}}

We illustrate the use of the moment-matching strategy by implementing it within location-scale NRMI mixture models, which can be represented in hierarchical form as
\begin{align*}
Y_i | \mu_i,\sigma_i & \simind k( \cdot |
\mu_i,\sigma_i),\quad i=1,\ldots,n,
\nonumber
\\
(\mu_i,\sigma_i)|\tilde P & \simiid \tilde P,\quad i=1,\ldots,n,\nonumber
\\
\tilde P & \sim \text{NRMI},\nonumber
\end{align*}
where $k$ is a kernel parametrized by $(\mu,\sigma)\in\R\times\R_+$ and the NRMI $\tilde P$ is defined in~\eqref{eq:NRMI}. Under this framework, density estimation is carried out by evaluating the posterior predictive density. Specifically, we consider the Gaussian kernel $k(x|\mu,\sigma) = \mathcal{N}(x|\mu,\sigma)$ and NGG on locations and scales with a normal base measure $P_0$, parameter $\theta=1$ in Equation~\eqref{eq:ngg}, and varying stability parameter $\gamma \in\{0,0.25,0.5,0.75\}$.

The dataset we consider is the popular \texttt{Galaxy} dataset, which consists of velocities of $82$ distant galaxies diverging from our own galaxy. Since the data are clearly away from zero (range from 9.2 to 34), Gaussian kernels, although having the whole real line as support, are typically employed in its analysis.

As far as the simulation algorithm is concerned, based on Sections~\ref{sec:prior} to~\ref{sec:ibp}, the following moment-matching Ferguson \& Klass posterior sampling strategy is implemented: (1) evaluate the threshold $M(\ell)$ which validates trajectories of the CRM using Algorithm~\ref{algo:FKa} on the prior distribution; (2) implement Algorithm~\ref{algo:FKa} on the posterior distribution using the threshold $M(\ell)$. More elaborate and suitably tailored moment-matching strategies can be devised for specific models. However, to showcase the generality and simplicity of our proposal we do not pursue this here.

In particular, we set
$\ell_M = 0.01$. We compare the output to the \FKa with heuristic relative error $e_M$ criterion, which consists of step (2) only with truncation dictated by the relative error for which we set $e_M \in\{ 0.1, 0.05, 0.01\}$. For both algorithms the Gibbs sampler is run for $20,000$ iterations with a burn-in of $4,000$, thinned by a factor of $5$.

In order to compare the results, we compute the Kolmogorov--Smirnov distance $d_{KS}(\hat F_{\ell_M}, \hat F_{e_M})$ between associated estimated cumulative distribution functions (cdf) $\hat F_{\ell_M}$ and $\hat F_{e_M}$ under, respectively, the moment-match and the relative error criteria. The results are displayed in Table~\ref{tab:galaxy}. The estimated cdf $\hat F_{\ell_M}$ with $\ell_M = 0.01$ can be seen as a reference estimate since the truncation error is controlled uniformly across the different values of $\gamma$ by the moment-match at the CRM level. First, one immediately notes that the smaller $e_M$, the closer the two estimates become (in the $d_{KS}$ distance). Second, and more importantly, the numerical values of the distances heavily depend on the particular choice of the parameter $\gamma$ for any given $e_M$. In fact, $\hat F_{\ell_M}$ and $\hat F_{e_M}$ are significantly further apart for large values of $\gamma$ than for small ones. This clearly shows that the quality of approximation with the heuristic criterion of the relative index is highly variable in terms of a single parameter; in passing from $\gamma=0$ to $\gamma=0.75$ the distance increases by at least a factor of $2$. This means that for comparing correctly CRM based models with different parameters one would need to pick different relative indices for each value of the parameter. However, there is no way to guess such thresholds without the guidance of an analytic criterion. And, this already happens by varying a single parameter, let alone when changing CRMs for which the same $e_M$ could imply drastically different truncation errors. This seems quite convincing evidence supporting the abandonment of heuristic criteria for determining the truncation threshold and the adoption of principled approaches such as the moment-matching criterion proposed in this paper.

\begin{table}[ht!]
\begin{center}
\begin{tabular}{lccc}
\toprule
$\gamma$& $e_M = 0.1$ & $e_M = 0.05$ & $e_M = 0.01$ \\
  \hline 
$0$ & 19.4 & 15.5 & 9.2 \\
$0.25$ & 31.3 & 23.7 & 15.1 \\
$0.5$ & 42.4  & 28.9 & 18.3 \\
$0.75$ & 64.8 & 41.0 & 23.2 \\
\bottomrule
\end{tabular}
\end{center}
\caption{\label{tab:galaxy} \texttt{Galaxy} dataset. Kolmogorov--Smirnov distance $d_{KS}(\hat F_{\ell_M}, \hat F_{e_M})$ between estimated cdfs $\hat F_{\ell_M}$ and $\hat F_{e_M}$ under, respectively, the moment-match (with $\ell_M = 0.01$) and the relative error (with $e_M = 0.1, 0.05, 0.01$) criteria. The mixing measure of normal mixture is the \NGG with varying $\gamma\in\{0,0.25,0.5,0.75\}$.}
\end{table}

\appendix
\section{Proof of Proposition 1} 
For any measurable set $A$ of $\X$, the $n$-th moment of $\tilde \mu(A)$, if it exists, is given by $m_n(A) = (-1)^nL_A^{(n)}(0)$, where $L_A^{(n)}(0)$ denotes the $n$-th derivative of the Laplace transform $L_A$ in \eqref{eq:lapl_transform} evaluated at 0. The result is proved by applying Fa\`a di Bruno's formula to~\eqref{eq:lapl_transform} for obtaining the derivatives. 
\section{Evaluation of the tail sum of the \SBP}\label{sec:app_tail}
Here we provide an evaluation of the tail sum~\eqref{eq:tail_sum} in the case of the \SBP.
We start by stating a lemma useful for upper bounding the tail sum.
\begin{lemma}\label{lem:N} Let function $N(\,\cdot\,)$ be as in~\eqref{eq:funcM_BP} for the \SBP. Then for any $\xi>0$
\begin{equation*}
N^{-1}(\xi) \leq  \left\{
\begin{matrix}
\edr^{\frac{1-\xi/a}{c}} \qquad \text{if } \sigma=0,\\
(\alpha\xi+\beta)^{-1/\sigma} \quad \text{if } \sigma\in(0,1),
\end{matrix}
\right.
\end{equation*}
where $\alpha = \sigma\Gamma(1-\sigma)\frac{\Gamma(c+\sigma)}{a\Gamma(c+1)}$ and $\beta = 1-\frac{\sigma}{c+\sigma}\Gamma(1-\sigma)$.
\end{lemma}
\begin{proof}
For  $\sigma=0$, from $u^{-1}(1-u)^{c-1}\leq u^{-1}+(1-u)^{c-1}$ one obtains $\int_{v}^{1} u^{-1}(1-u)^{c-1}\ddr u\leq 1/c -\log v$. Hence, $N(v)/a\leq 1-c\log v$ and $N^{-1}(\xi)\leq \edr^{(1-\xi/a)/c}$. The argument for $\sigma\neq 0$ follows along the same lines starting from  $u^{-1-\sigma}(1-u)^{\sigma+c-1}\leq \Gamma(1-\sigma)u^{-\sigma-1}+(1-u)^{\sigma+c-1}$.
\end{proof}
\begin{prop}\label{prop:proba}
Let $(\xi_j)_{j\geq 1}$ be the jump times for a homogeneous Poisson process on $\R^+$ with unit intensity. Define the tail sum of the \SBP as
\begin{equation*}
T_M = \sum_{j=M+1}^\infty N^{-1}(\xi_j),
\end{equation*}
where  $N(\,\cdot\,)$ is given by~\eqref{eq:funcM_BP}. Then for any $\epsilon\in(0,1)$,
\begin{equation*}
\Proba\Big(T_M\leq t_M^\epsilon\Big) \geq 1-\epsilon, \text{ for }
t_M^\epsilon = \left\{
\begin{matrix}
\frac{C_1}{\epsilon}\edr^{\frac{1}{c}-\frac{\epsilon M}{C_1}} \qquad \qquad \quad \text{if } \sigma=0,\\
\frac{\sigma}{1-\sigma}\frac{(C_2\epsilon)^{1/\sigma}}{(M+\beta C_2/\epsilon)^{1/\sigma-1}}  \quad \text{if } \sigma\in(0,1),
\end{matrix}
\right.
\end{equation*}
where $C_1=2ac\edr$ and $C_2=2\edr/\alpha$ do not depend on $\epsilon$.
\end{prop}
\begin{proof}
The proof follows along the same lines as the proof of Theorem A.1. in \citet{brix1999generalized}. Let $q_j$ denote the $\epsilon2^{M-j}$ quantile, for $j=M+1,M+2,\ldots$, of a gamma distribution with mean and variance equal to $j$. Then
\begin{equation*}
\Proba\bigg(\sum_{j=M+1}^\infty N^{-1}(\xi_j)\leq \sum_{j=M+1}^\infty N^{-1}(q_j)\bigg)\geq 1-\epsilon.
\end{equation*}

An upper bound on $\tilde t_M^\epsilon=\sum_{j=M+1}^\infty N^{-1}(q_j)$ is then found by resorting to Lemma~\ref{lem:N} along with the inequality $q_j\geq \frac{\epsilon}{2\edr}j$. If $\sigma=0$
\begin{equation*}
\tilde t_M^\epsilon \leq \edr^{1/c} \sum_{j=M+1}^\infty \edr^{-\frac{q_j}{ac}} \leq\edr^{1/c} \sum_{j=M+1}^\infty \edr^{-\frac{\epsilon j}{2ac\edr }}\leq\edr^{1/c} \frac{2ac\edr }{\epsilon}\edr^{-\frac{\epsilon M}{2ac\edr }},
\end{equation*}
whereas if $\sigma\neq 0$
\begin{equation*}
\tilde t_M^\epsilon\leq\sum_{j=M+1}^\infty (\alpha q_j+\beta)^{-\frac{1}{\sigma}} \leq\sum_{j=M+1}^\infty \left(\frac{\alpha \epsilon j}{2\edr}+\beta\right)^{-\frac{1}{\sigma}} =
\left(\frac{2\edr}{\alpha\epsilon}\right)^{-\frac{1}{\sigma}}\sum_{j=M+1}^\infty \left(j+\frac{2\edr\beta}{\alpha\epsilon}\right)^{-\frac{1}{\sigma}}.
\end{equation*}

The result follows by bounding the last sum by $\int_{M}^{\infty}\left(x+\frac{2\edr\beta}{\alpha\epsilon}\right)^{-\frac{1}{\sigma}}\ddr x$.
\end{proof}

The bound $t_M^\epsilon$ obtained in Proposition~\ref{prop:proba} is exponential when $\sigma=0$ and polynomial when $\sigma\neq 0$, but it is very conservative as already pointed out by \citet{brix1999generalized}. This finding is further highlighted in the table associated to Figure~\ref{fig:proba}, where the bound $t_M^\epsilon$ is computed with appropriate constants derived from the proof. In contrast, the bound $\tilde t_M^\epsilon$ obtained by direct calculation of the quantiles $q_j$ (instead of resorting to a lower bound on them) is much sharper. Figure~\ref{fig:proba} displays the sharper bound $\tilde t_M^\epsilon$. Inspection of the plot demonstrates a decrease pattern in this bound in probability which is reminiscent of the ones for the indices $\ell_M$ and $e_M$ studied in the paper. This observation is a further indication that the \FKa is a tool with well-behaved approximation error.

\begin{center}
\begin{figure}[ht!]
\begin{minipage}{7cm}
\includegraphics[width=\linewidth]{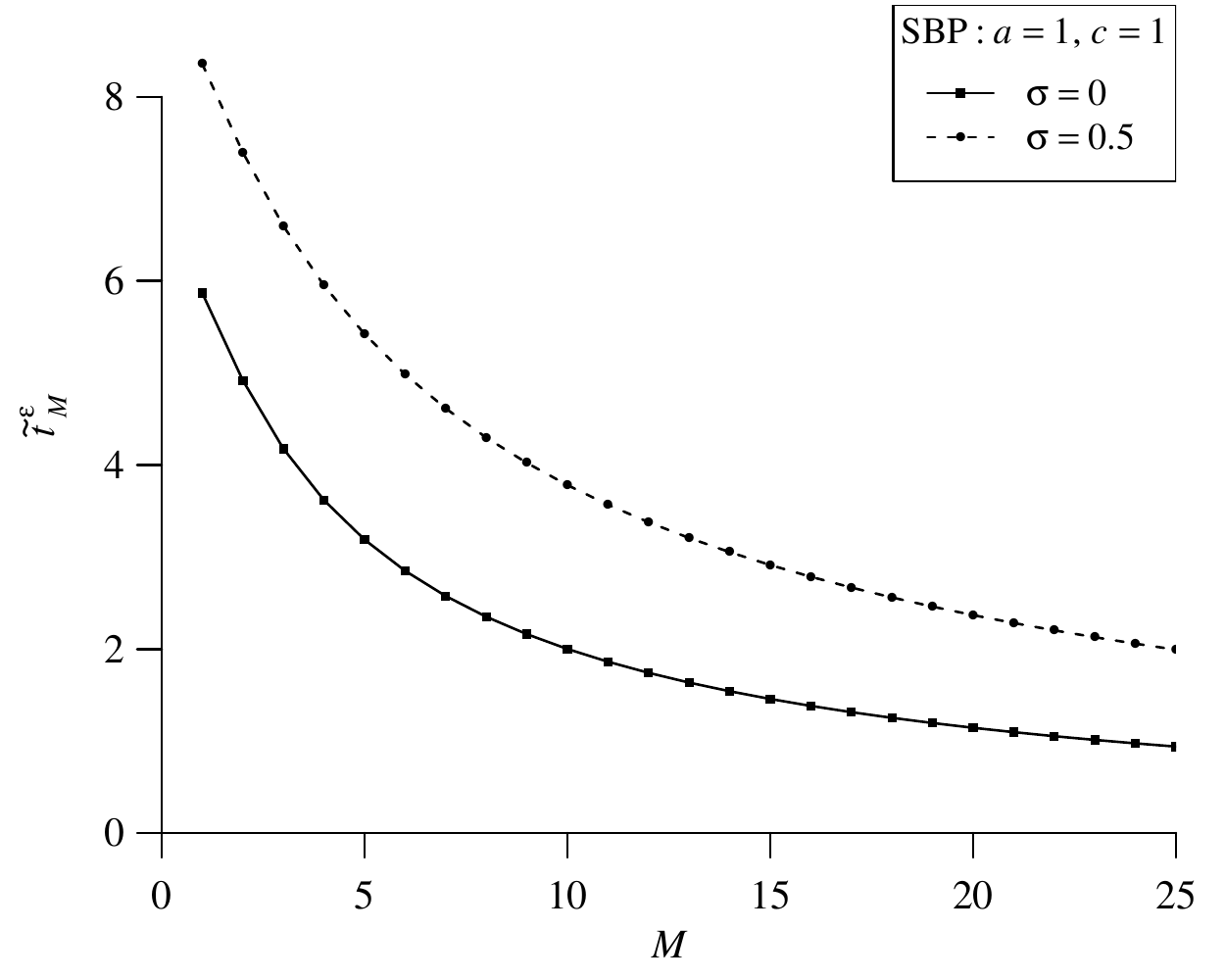}
\end{minipage}
\hskip1cm
\begin{minipage}{5cm}
\begin{tabular}{ccccc}
\toprule
& $M$ & 25 & 100 & 500 \\\hline
\multirow{2}{*}{$t_M^\epsilon$} & $\sigma=0$ & 1411 & 1230 & 589 \\
\cline{2-5}
& $\sigma=0.5$ &  1554 & 1250 & 612 \\\hline
\multirow{2}{*}{$\tilde t_M^\epsilon$} & $\sigma=0$ &   0.942 & 0.230 & 0.003 \\
\cline{2-5}
& $\sigma=0.5$ &   1.998 & 0.534 & 0.008  \\\bottomrule
\end{tabular}
\end{minipage}
\caption{Stable-beta process with parameters $\sigma=0$ and $\sigma=0.5$. Left: Bound in probability $\tilde t_M^\epsilon$ of the tail sum $T_M$ obtained by direct calculation of the quantiles $q_j$ with $\epsilon=10^{-2}$ as the truncation level $M$ increases.
Right:  Bounds $t_M^\epsilon$ (provided in Proposition~\ref{prop:proba}) and $\tilde t_M^\epsilon$ (obtained by direct calculation of the quantiles $q_j$) of the tail sum after $M$ jumps with $\epsilon=10^{-2}$.}
\label{fig:proba}
\end{figure}
\end{center}

\bibliographystyle{apalike}
\bibliography{biblio}

\end{document}